\documentclass[letterpaper, onecolumn, notitlepage, longbibliography, floatfix, superscriptaddress, aps, unpublished, 12pt,accepted=2023-08-31]{quantumarticle}
\pdfoutput=1
\usepackage{amsmath,amsfonts, amssymb, amsthm, dsfont}
\usepackage{bm}
\usepackage{longtable, array, booktabs}
\usepackage[table]{xcolor}
\usepackage{mathrsfs}
\usepackage{graphicx}
\usepackage{verbatim}
\usepackage{hyperref}
\usepackage{wasysym}
\usepackage{xcolor}
\usepackage{tikz}
\usepackage{physics}

\usepackage[percent]{overpic}
\usepackage{floatrow}
\usepackage[caption=false,label font={bf,normalsize}]{subfig}
\newtheorem{lemma}{Lemma}

\newtheorem{corollary}{Corollary}

\newcommand{\be}{\begin{equation}}
\newcommand{\ee}{\end{equation}}
\definecolor{quantumpurple}{RGB}{82, 35, 124}

\hypersetup{colorlinks=true,citecolor=quantumpurple,linkcolor=quantumpurple, urlcolor=quantumpurple}

\usepackage{empheq}
\usepackage{xcolor}
\definecolor{lightgreen}{HTML}{90EE90}
\newcommand{\boxedeq}[2]{\begin{empheq}[box={\fboxsep=6pt\fbox}]{align}\label{#1}#2\end{empheq}}

\begin{document}
\title{\textcolor{quantumpurple}{Quantum error correction with fractal topological codes}}

\author{Arpit Dua}
\thanks{adua@caltech.edu}
\affiliation{Department of Physics and Institute for Quantum Information and Matter, California Institute of Technology, Pasadena, CA 91125 USA}

\author{Tomas Jochym-O'Connor}
\thanks{tjoc@ibm.com}
\affiliation{IBM Quantum, IBM T.J. Watson Research Center, Yorktown Heights, NY 10598 USA}
\affiliation{IBM Almaden Research Center, San Jose, CA 95120 USA}

\author{Guanyu Zhu}
\thanks{guanyu.zhu@ibm.com}
\affiliation{IBM Quantum, IBM T.J. Watson Research Center, Yorktown Heights, NY 10598 USA}
\affiliation{IBM Almaden Research Center, San Jose, CA 95120 USA}

\begin{abstract}
\noindent
Recently, a class of fractal surface codes (FSCs), has been constructed on fractal lattices with Hausdorff dimension $2+\epsilon$, which admits a fault-tolerant non-Clifford CCZ gate \cite{zhu2021topological}. We investigate the performance of such FSCs as fault-tolerant quantum memories. We prove that there exist decoding strategies with non-zero thresholds for bit-flip and phase-flip errors in the FSCs with Hausdorff dimension $2+\epsilon$. For the bit-flip errors, we adapt the sweep decoder, developed for string-like syndromes in the regular 3D surface code, to the FSCs by designing suitable modifications on the boundaries of the holes in the fractal lattice. Our adaptation of the sweep decoder for the FSCs maintains its self-correcting and single-shot nature. For the phase-flip errors, we employ the minimum-weight-perfect-matching (MWPM) decoder for the point-like syndromes. We report a sustainable fault-tolerant threshold ($\sim 1.7\%$) under phenomenological noise for the sweep decoder and the code capacity threshold (lower bounded by $2.95\%$) for the MWPM decoder for a particular FSC with Hausdorff dimension $D_H\approx2.966$. The latter can be mapped to a lower bound of the critical point of a  confinement-Higgs transition on the fractal lattice, which is tunable via the Hausdorff dimension.
\end{abstract}
\maketitle

Topological stabilizer codes are a highly promising class of codes for scalable architectures of fault-tolerant quantum memories and quantum computation \cite{bravyi1998quantum,Kitaev03,Dennis_2002, Bombin_2006,Fowler_2012_large_scale}. This can be attributed to the geometrically local stabilizer terms and high fault-tolerant thresholds. However, the power of topological stabilizer codes is restricted. For instance, the Bravyi-K$\ddot{\text{o}}$nig bound provides an upper bound on the set of transversally implementable logical gates for topological stabilizer codes on $D$-dimensional lattices to be the $D$-th level of the Clifford hierarchy~\cite{Bravyi_2013_PRL}. Hence, in two-dimensional (2D) topological stabilizer codes, only the Clifford group can be directly implemented fault-tolerantly~\cite{JOKY18}; in order to have a universal logical gate set, magic-state distillation is required, leading to additional space-time overhead~\cite{bravyi2005, Fowler_2012_large_scale, litinski19}. One potential solution is resorting to the 2D non-Abelian Turaev-Viro codes, which evades the simple stabilizer formalism \cite{levin2005,Koenig_2010,schotte2020quantum, Zhu:2020_constant_depth, Lavasani2019universal, Zhu:2018CodeLong, Zhu:2017tr}. An alternate proposal is the three-dimensional (3D) topological stabilizer code for which the fault-tolerant gates can reach the third level of the Clifford hierarchy. In particular, there exists a fault-tolerant single-shot implementation of the non-Clifford CCZ gate in the 3D surface code~\cite{KYP15, VB19}, analogous to the implementation of the $T$~gate in the 3D color code~\cite{Bombin15a,Bombin15b}. Another advantage of the 3D surface code is that the bit-flip sector of syndromes is string-like for which there exists a cellular automaton-based decoding strategy with a local update rule, called the sweep decoder~\cite{kubica_ca, vasmer2020cellular}. This shows that the system is self-correcting for bit-flip errors ~\cite{Brown_ROMP_2016}, implying the single-shot error correction property \cite{Bombin15b}. This remedies the overhead associated with non-local classical communication of the usual decoder \cite{Fowler_PRL2012} in realistic scenarios and also allows hardware-efficient autonomous error correction with dissipation engineering \cite{memories_diss_2011} for bit-flip errors. 

In recent work, the authors have proved that topological order can exist on fractal lattices embedded in $D$~spatial dimensions. In particular, a family of topological stabilizer codes on fractal lattices, embedded in $D$ dimensions, with Hausdorff dimensions $D_H = D-1+\epsilon$, where $\epsilon$ could be an arbitrarily small nonzero number, is constructed~\cite{zhu2021topological}. 
We refer to this family of codes on fractal lattices as the fractal surface codes (FSCs)~\cite{zhu2021topological}. FSCs in 3D, which we simply refer to as FSCs from here on, are constructed by punching homotopically trivial holes with smooth boundaries in the 3D surface code such that the resulting lattice is a fractal. The code distance $d=L$ is preserved under such code puncturing, where $L$ is the linear system size. Moreover, the non-Clifford CCZ gate implementation of the 3D surface code is still possible up to some modifications at the hole boundaries \cite{zhu2021topological}. Surprisingly, due to this possibility, the space-overhead associated with the fault-tolerant universal gate set is reduced to $\mathcal{O}(d^{2+\epsilon})$ compared to $\mathcal{O}(d^3)$ for the 3D surface code. The FSCs still require 3D connectivity of qubits, which can be realized in architectures such as 3D integrated superconducting qubits \cite{mallek2021fabrication,2017_3DISC,IBM3D} and photonic qubits~\cite{bartolucci2021,bombin2021}.

In this work, we study the performance of the FSCs as fault-tolerant quantum memories. For FSCs with arbitrary Hausdorff dimension $D_H$$=$$2$$+$$\epsilon$, we prove that there exist decoding strategies with nonzero thresholds for both the string-like and point-like syndromes under the bit-flip and phase-flip errors respectively. Moreover, for a particular FSC with Hausdorff dimension $D_H$$=$$\ln 26/\ln 3 $$\approx $$2.966$, we report the fault-tolerant (sustainable) threshold using a variant of the sweep decoder for the string-like syndromes in the presence of pure bit-flip noise and the code capacity threshold using the minimum weight perfect matching (MWPM) decoder for the point-like syndromes in presence of pure phase-flip errors. We chose this $D_H$ due to numerical limitations; the decoders are proven to work for FSCs with lower Hausdorff dimension $D_H$$=$$2$$+$$\epsilon$ as well. The sweep decoder needs to be adapted to the presence of the holes in the fractal lattice. We modify the local update rule of the sweep rule near the hole boundaries such that the self-correcting and the single-shot properties of the decoder are maintained for any fractal dimension $D_H$$=$$2$$+$$\epsilon$. We demonstrate the sweep decoder performance for a particular FSC with $D_H$$\approx$$ 2.966$, as mentioned, by calculating logical failure rates under $N$ rounds of noisy syndrome measurements. In the limit $N$$\rightarrow$$\infty$, we obtain the so-called sustainable threshold. The intuition for why the modified sweep decoder maintains a threshold at fractal dimensions is due to the distribution and spacing of the holes in the lattice with a scale-invariant pattern. Namely, while the presence of holes may delay the speed at which a particular error is cleaned up, the holes do not affect the correctability of the various connected components of errors.

\section{Fractal surface code}
In order to define the FSC, we start with the 3D surface code. The 3D surface code is defined on the cubic lattice $\mathcal{L}$, where qubits sit on the edges of the lattice. The code is defined by the following stabilizer generators, $A_\mathsf{v}=\prod_{\mathsf{e} \ni \mathsf{v} }X_\mathsf{e}$ and $B_\mathsf{f}=\prod_{ \mathsf{e} \in \mathsf{f} }Z_\mathsf{e}$ where $\mathsf{v}$, $\mathsf{e}$ and $\mathsf{f}$ denote the vertices, edges and faces (0-, 1-, 2-cells) of $\mathcal{L}$, and $Z_\mathsf{e}$ and $X_\mathsf{e}$ are Pauli operators associated with a given qubit sitting on the edge~$\mathsf{e}$. Pictorially, the bulk stabilizer generators, $A_\mathsf{v}$ and $B_\mathsf{f}$ up to their translates are as follows,  
\be
  \includegraphics[width=0.7\columnwidth]{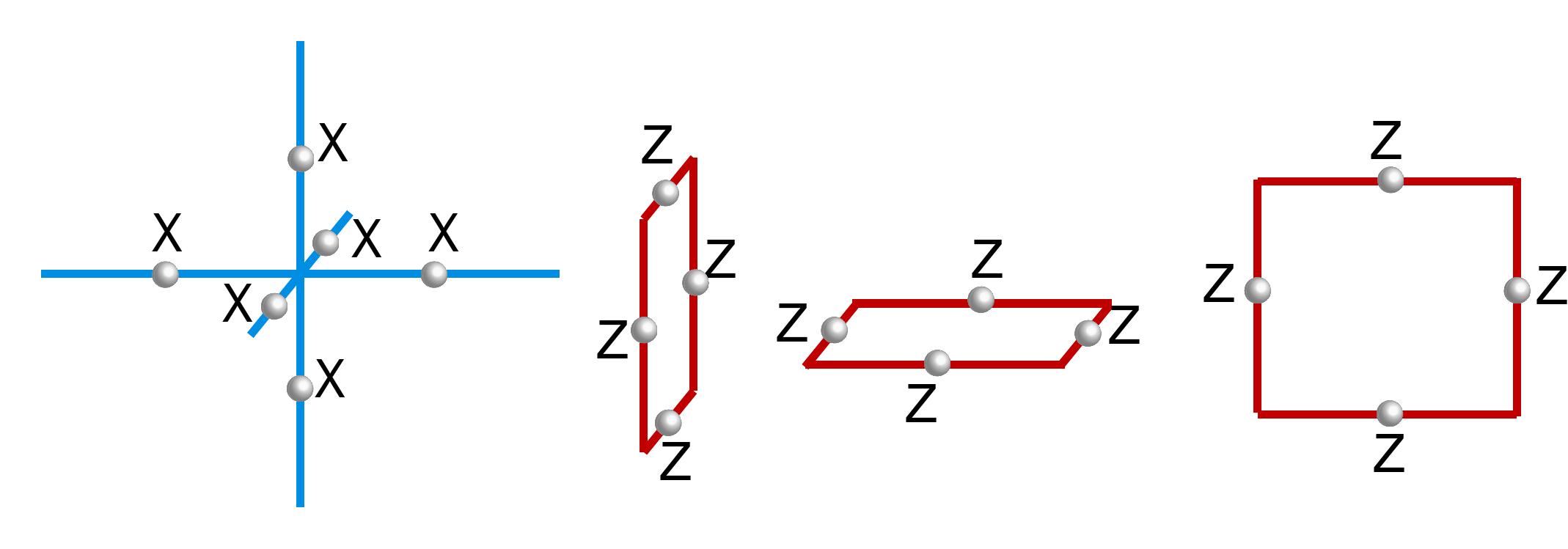}.
\ee
We also consider the dual lattice $\mathcal{L}^*$ obtained from the original lattice $\mathcal{L}$ with the following conversion of the cells: $\mathsf{e} \leftrightarrow \mathsf{f}$ and $\mathsf{v} \leftrightarrow \mathsf{c}$, where $\mathsf{c}$ denotes a cube (3-cell). See Fig.~\ref{fig:lattice_stabs}(a) for these stabilizers on both the original and dual lattice and the latter will be useful to describe the sweep decoder. Since the qubits now sit on the faces $\mathsf{f}^*$ of the dual lattice $\mathcal{L^*}$, the vertex stabilizers are supported on cubes $\mathsf{c}^*$ of $\mathcal{L^*}$ while the plaquette stabilizers are associated to the edges $\mathsf{e}^*$ of $\mathcal{L^*}$ such that the support is on faces $\mathsf{f}^*$ neighboring edge $\mathsf{e}^*$. Mathematically, $A_\mathsf{v}\equiv A_{\mathsf{c}^*}=\prod_{\mathsf{f}^*\in \mathsf{c}^*}X_{\mathsf{f}^*}$ and $B_\mathsf{f}\equiv B_{\mathsf{e}^*}=\prod_{\mathsf{f}^*\ni \mathsf{e}^*}Z_{\mathsf{f}^*}$. The violations of $X$-stabilizers $A_\mathsf{v}$ on $\mathcal{L}$ give rise to the $e$-excitations or point-like syndromes $S_0$ as boundaries of string-like $Z$ error chain $E_1$ (1-chain), \textit{i.e.}, $S_0=\partial E_1$; those of $Z$-stabilizers $B_{\mathsf{e}^*}$ on $\mathcal{L}^*$ lead to $m$-excitations or string-like syndromes $S^*_{1}$ as boundaries of membrane-like $X$ error chain $E^*_2$ (2-chain), \textit{i.e.}, $S^*_{1}=\partial E^*_2$. The associated pair of anti-commuting logical operators is given by (1) a string-like operator $\overline{Z}=\prod_{\mathsf{e} \in [c_1]} Z_\mathsf{e}$ supported on an equivalent class of 1-cycles on $\mathcal{L}$ belonging to the 1st relative homology group over $\mathbb{Z}_2$ coefficients, \textit{i.e.}, $[c_1] \in H_1(\mathcal{L}, \mathcal{B}_e; \mathbb{Z}_2)$, which describes the class of non-contractible 1-cycles including the absolute cycles (loops) and relative cycles (open strings) terminated on the \textit{rough boundaries} (with dangling edges), also called $e$-boundaries $\mathcal{B}_e$ since they condense the $e$-excitations, as illustrated in Fig.~\ref{fig:lattice_stabs}(b); (2) a membrane-like operator $\overline{X}=\prod_{\mathsf{f}^* \in [c^*_2]} X_{\mathsf{f}^*}$ supported on an equivalent class of 2-cycles on $\mathcal{L}^*$ belonging to the 2nd relative homology group, \textit{i.e.}, $[c_2^*] \in H_2(\mathcal{L}^*, \mathcal{B}^*_m; \mathbb{Z}_2)$, which describes the class of non-contractible  absolute 2-cycles (closed membranes) and relative 2-cycles (open membranes) terminated on the \textit{smooth boundaries} (without  dangling edges), also called $m$-boundaries $\mathcal{B}_m^*$ since they condense the $m$-excitations. 

\begin{figure}[H]
    \centering
 \includegraphics[scale=0.4]{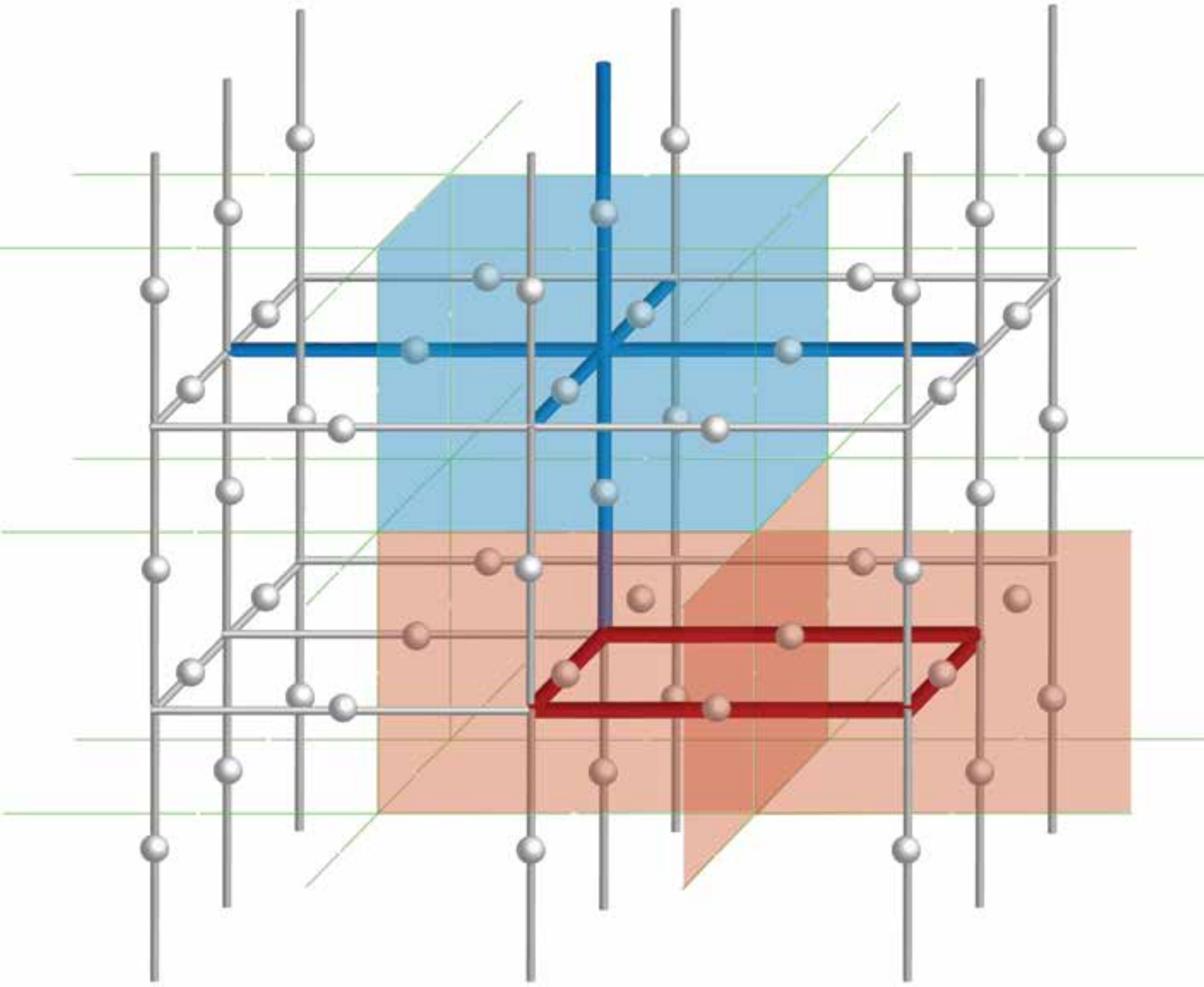}
 \includegraphics[scale=0.4]{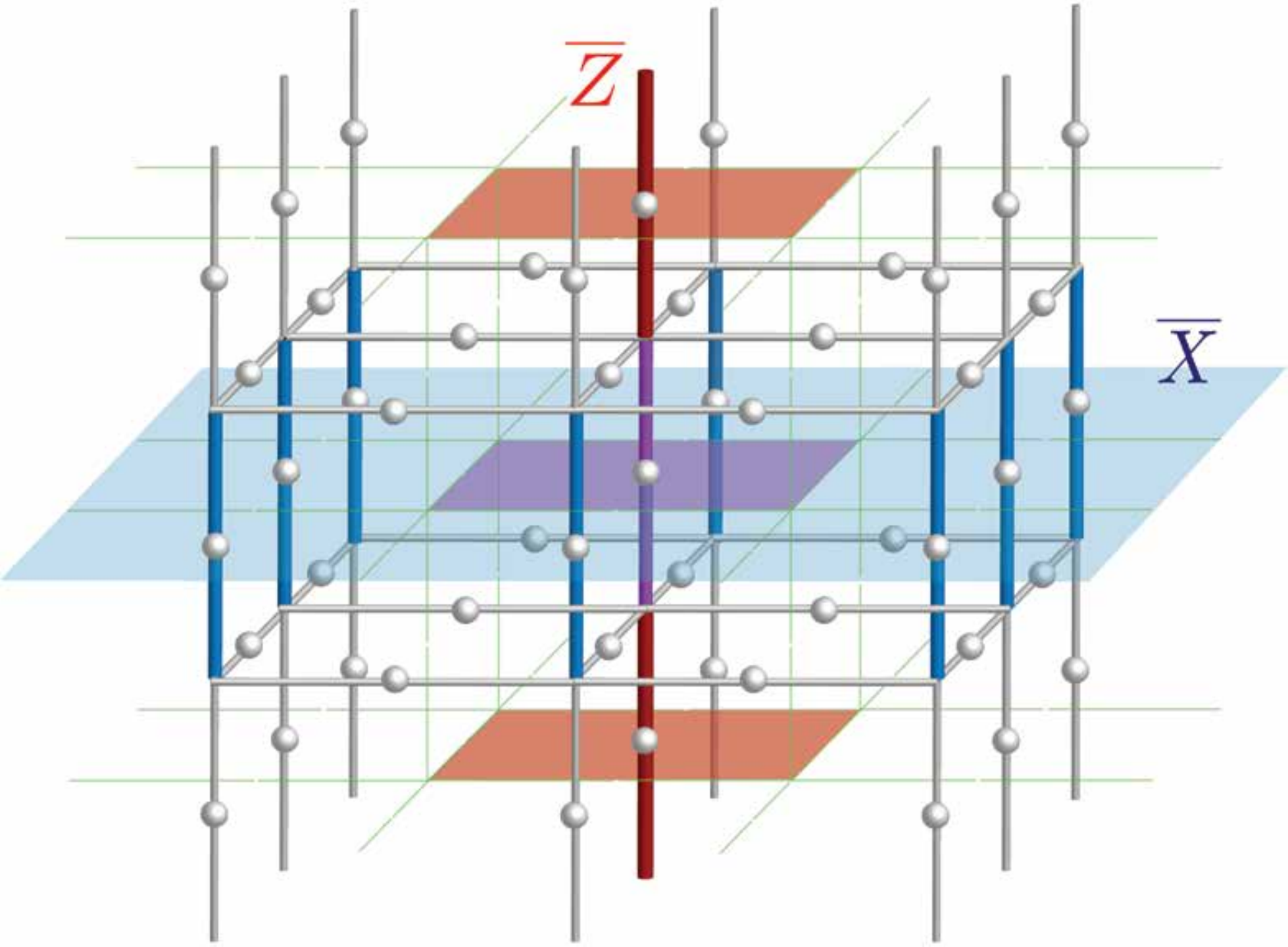}     
    \caption{Left: 3D surface code on the original lattice $\mathcal{L}$ (black) and the dual lattice $\mathcal{L}^*$ (thin green). The edges highlighted in red (red faces on $\mathcal{L}^*$) denote the action of Pauli $X$ operators, while the edges highlighted in blue (blue faces on $\mathcal{L}^*$) denote the action of Pauli $Z$ operators. The vertex stabilizer $A_\mathsf{v}=\prod_{\mathsf{e} \ni \mathsf{v} }X_\mathsf{e}$ and plaquette stabilizers $B_\mathsf{f}=\prod_{ \mathsf{e} \in \mathsf{f} }Z_\mathsf{e}$, where $e$ and $f$ refer to the edges and faces of $\mathcal{L}$, are shown in blue and red respectively. Right: The string logical operator $\overline{Z}$ and the membrane logical operator $\overline{X}$ are shown.}
    \label{fig:lattice_stabs}
\end{figure}
The FSC is obtained from the 3D surface code on a cubic lattice by punching holes with $m$-boundaries which we refer to as $m$-holes, in an iterative manner. The first iteration involves starting with the original cubic lattice of the surface code, which we call the level-$0$ cube. In the $\ell^\text{th}$ iteration, one divides each level-$\ell$ cube equally into $a \times a \times a$ level-$(\ell+1)$ cubes with linear size $1/a$ of a level-$\ell$ cube, and punch an $m$-hole in the center occupying $b \times b\times b$ cubes, and we obtain the Fractal Cube (FC) geometry for level $\ell$ as $FC(a, b, \ell)$. The fractal cube geometry is generated in the asymptotic limit, \textit{i.e.}, $FC(a, b) \equiv \lim_{\ell \rightarrow \infty} FC(a, b, \ell)$. Requiring $b<a$, we get an asymptotic fractal cube geometry with Hausdorff dimension $D_H\rightarrow2$ in the limit $b/a\rightarrow 1$~\cite{zhu2021topological}. 

\begin{figure}[H]
    \centering
   \sidesubfloat[]{\includegraphics[width=0.82\columnwidth]{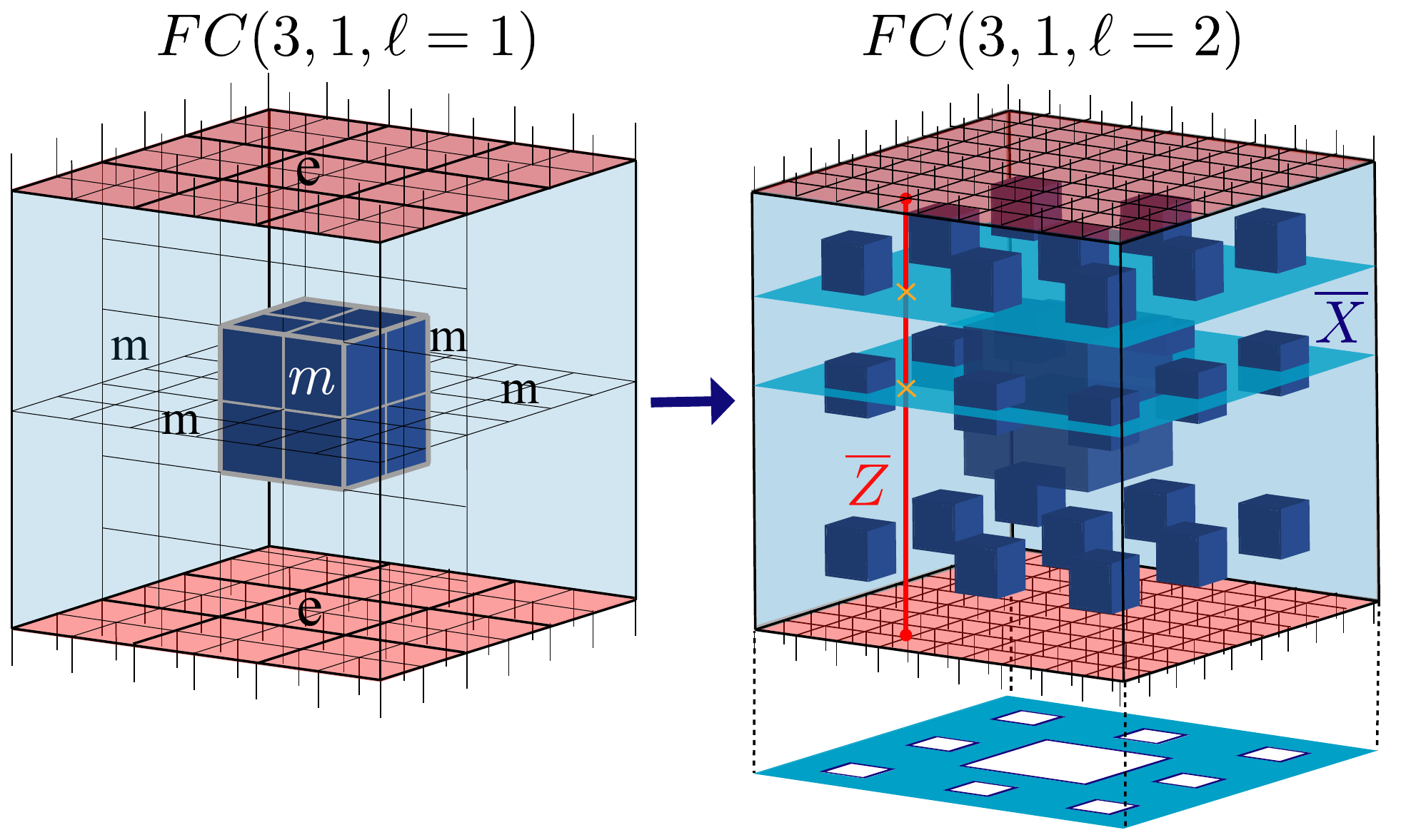}}\\
   \sidesubfloat[]{ \includegraphics[scale=0.5]{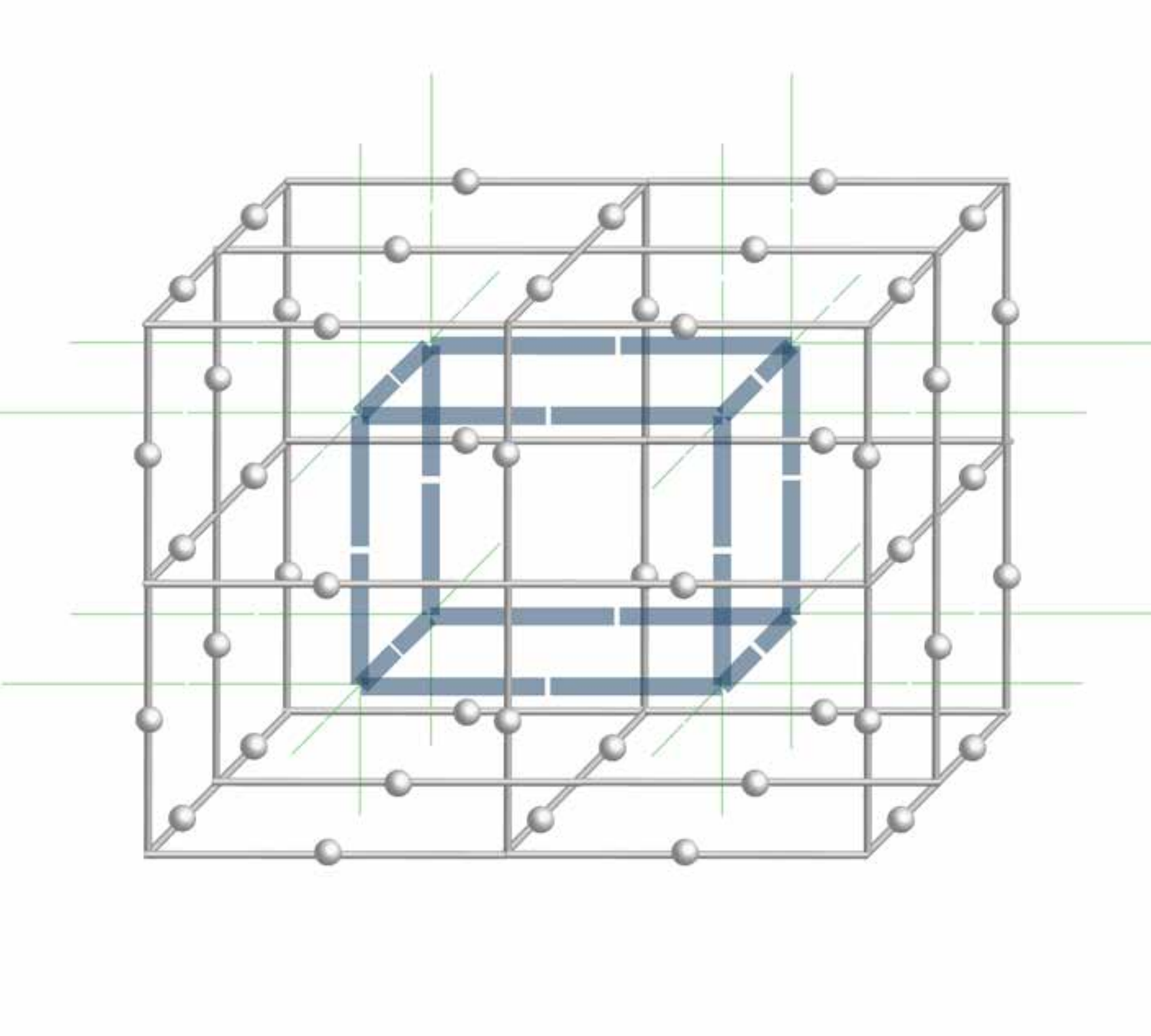}}
   \vspace{-3mm}
    \caption{(a) Levels 1 and 2 of the fractal surface code are shown as $FC(3,1,l=1)$ and $FC(3,1,l=2)$. For level 1, the single $m$-hole is shown, and for level 2, an extra iteration of punched holes is shown. The anti-commuting pair of logical operators $\overline{X}$ and $\overline{Z}$ are also shown for level 2. The minimum weight representation of membrane operator $\overline{X}$, which is interrupted by the holes and forms a Sierpinski carpet, is shown separately at the bottom.  In (b), we show a $1\times 1\times 1$ size $m$-hole. On the original lattice, the missing edges form the hole, while on the dual lattice, the $m$-hole is shown as a cube with gray-blue edges, where the missing qubits are given by the faces of the cube.}
    \label{fig:FSC_and_mholes}
\end{figure}

The code space $\mathcal{H}_\mathcal{C}$ is mathematically determined by the 1st relative homology group $H_1(\widetilde{\mathcal{L}}, \mathcal{B}_e; \mathbb{Z}_2)$, defined on the fractal lattice $\widetilde{\mathcal{L}}$. As mentioned, the $m$-holes in the FSC are chosen to be homeomorphic to 3-dimensional balls, and they have a trivial contribution to the homology group, meaning that they do not encode any additional logical qubit (see Sec. II in Supplementary for a counting argument). This can be seen by the fact that logical-$Z$ string can neither terminate on the boundaries of m-holes (due to the condensation property) nor enclose these holes since they are 3-balls. Thus, the code space can be expressed as
\be\label{eq:homology_group}
\mathcal{H}_\mathcal{C}= \mathbb{C}^{H_1(\widetilde{\mathcal{L}}, \mathcal{B}_e; \mathbb{Z}_2)} = \mathbb{C}^{\mathbb{Z}_2}= \mathbb{C}^2,
\ee
where the non-trivial contribution of  $\mathbb{Z}_2$ to $H_1$ comes from the logical-$Z$ string connecting the top and bottom $e$-boundaries circumventing any $m$-holes (see Ref.~\onlinecite{zhu2021topological} for a detailed proof). We hence have one encoded logical qubit. The unique class of the dual logical-$X$ membranes, determined by the 2nd relative homology group which is isomorphic to the above 1st relative homology group~\footnote{This is due to the Poincaré-Lefschetz duality, which is also consistent with the structure of a  CSS code.}, is hence $H_2(\widetilde{\mathcal{L}}^*, \mathcal{B}^*_m; \mathbb{Z}_2) $$\cong$$H_1(\widetilde{\mathcal{L}}, \mathcal{B}_e; \mathbb{Z}_2) $$=$$ \mathbb{Z}_2$. The minimum weight logical-$X$ membrane now terminates on the $m$-holes and is supported on a Sierpinski carpet with Hausdorff dimension $D^X_H=\ln 8/\ln 3\approx 1.893$ and corresponding $X$-distance $d_X \sim L^{1.893}$ in the case of $FC(3,1)$, as illustrated in Fig.~\ref{fig:FSC_and_mholes}. The overall distance is determined by the minimum length of the logical-$Z$ string $d$$=$$\min\{d_X,d_Z\}$$=$$d_Z$$=$$L$. Since the $X$ error is still membrane-like with string-like syndromes on its boundary, the self-correction and single-shot nature of error correction for the loop sector are still expected to survive in the case of imperfect syndrome measurements.

\subsection*{Decoding the FSC}
We study the performance of the 3D FSC under two i.i.d.~Pauli noise models, \textit{i.e.}, the bit-flip and the phase-flip error models. The first one is described by the following Pauli noise channel, 
\begin{equation}\label{eq:channel}
\rho\rightarrow \sum_\mathsf{e} \left[ p_{X} X_\mathsf{e}\rho X_\mathsf{e}+(1-p_{X})\rho \right],
\end{equation}
where $\rho$ is the density matrix describing the state of the code and $p_{X}$ is the single-qubit $X$ error rate. To get the phase-flip noise model, switch $X$ to $Z$ in Eq.~\eqref{eq:channel}.

According to the asymptotic definition of the fractal cube geometry, the threshold of the FSC, $p_{\text{th}}(FC(a, b))$ is defined as
\begin{equation}
    p_{\text{th}}(FC(a, b))= \lim_{\ell \rightarrow \infty} p_{\text{th}}(FC(a, b, \ell))
\end{equation}
where $p_{\text{th}}(FC(a, b, \ell))$ is the threshold for level $\ell$ of the FSC. In this work, we evaluate the thresholds for levels $\ell=1,2$ and estimate the threshold for the limit $\ell\rightarrow\infty$ \textit{i.e.,} for $FC(a, b)$. 

As mentioned, we use a cellular automaton-based decoder called the sweep decoder to decode the $X$ errors and the MWPM decoder to decode the $Z$ errors. However, the sweep decoder needs to be adapted to account for the $m$-holes. In Sec.~\ref{sec:sweep} below, we first describe how the sweep decoder works for the 3D surface code and then generalize to the FSC. The sweep decoder can be defined generally for any causal codes \cite{vasmer2020cellular}. However, we restrict our description of the decoder to the 3D surface code on a cubic lattice for clarity. 

\section{Sweep Decoder}
\label{sec:sweep}
The main ingredient of the sweep decoder is the sweep rule. Intuitively speaking, the sweep rule takes into input the syndrome and applies a correction operator such that a part of the syndrome is swept towards a particular direction called the sweep direction. Eventually, under enough such sweeps, the syndrome is cleaned. For perfect measurements, the syndrome consists of closed loops; hence partial sections of a given loop move towards its remaining sections and close on to them. We illustrate the sweep rule in Fig.~\ref{fig:sweep_illust}(a). 

\begin{figure}[H]
    \centering
\includegraphics[width=0.95\columnwidth]{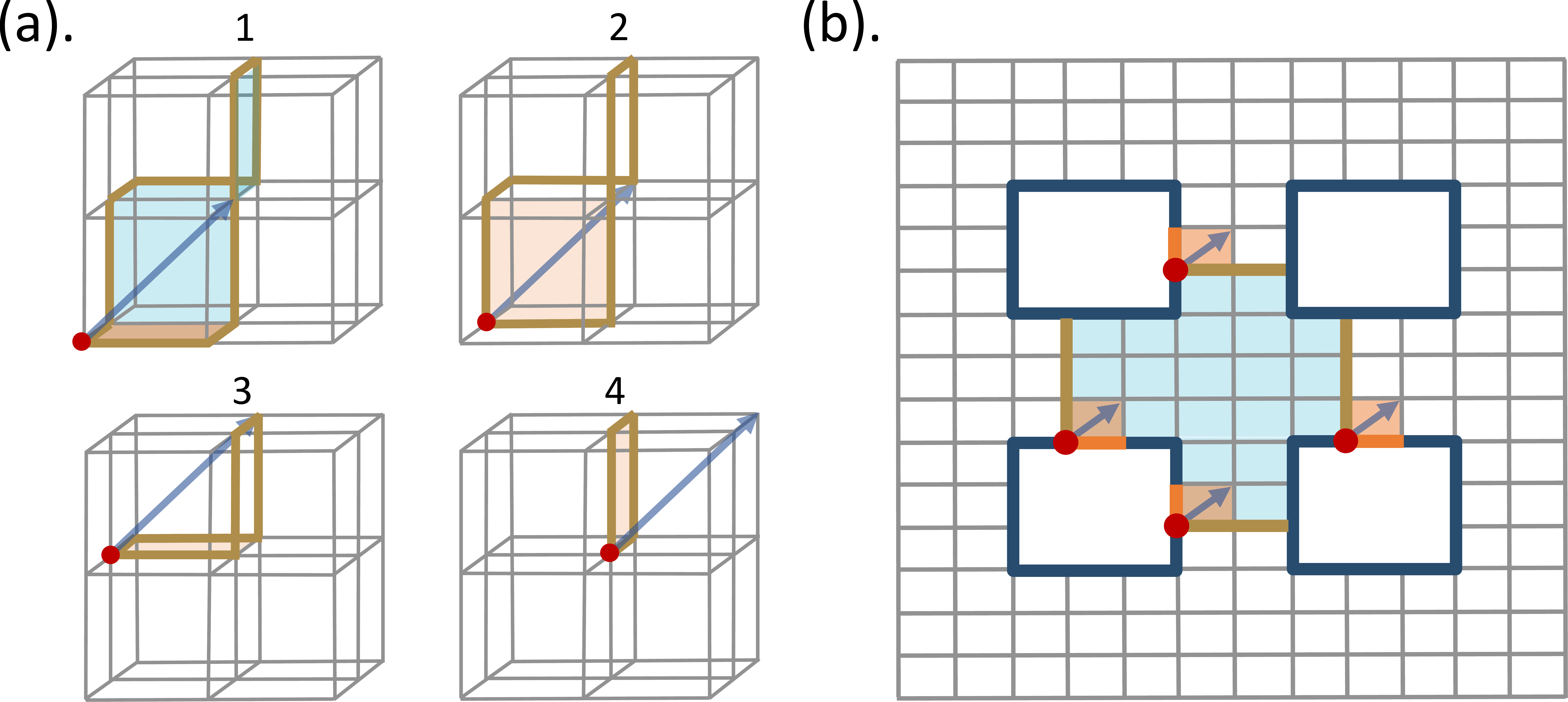}
    \caption{Illustration of the sweep rule. The sweep direction is shown using arrows next to extremal vertices in red. The plaquettes shaded in blue denote the errors that generate the syndrome, and the plaquettes shaded in orange are the ones that are swept or on which a correction is applied. (a) Sweep rule applied four times (1-4) on the syndrome (brown) that sits on the edges $\mathsf{e}^*$ of the dual cubic lattice $\mathcal{L}^*$. (b) A 2D projection of a trapped syndrome configuration (brown) between $m$-holes (dark blue). The modified sweep rule allows us to put imaginary syndromes on an edge (brown) in the future on the surface of the hole.}
    \label{fig:sweep_illust}
\end{figure}

We now state the sweep rule for the regular 3D surface code on a cubic lattice  $\mathcal{L}$ formally. To do so, it is convenient to consider the dual lattice $\mathcal{L}^*$ such that the syndromes live on the edges $\mathsf{e}^*$ and qubits on the faces $\mathsf{f}^*$. Note that from here on, we work in this dual picture. Based on the sweep direction $\hat{s}$, each edge attached to a vertex is assigned to be in the future or the past of that vertex. For example, consider a vertex $\mathsf{v}^*$ and for each edge $\mathsf{e}^*$ emanating from $\mathsf{v}^*$, associate a unit vector $\hat{\mathsf{v}}^*_{\mathsf{e}^*}$ in the direction along the edge $\mathsf{e}^*$ starting from $\mathsf{v}^*$. If the dot product of $\hat{s}$ with the $\hat{\mathsf{v}}^*_{\mathsf{e}^*}$ is non-negative, the edge is in the future of that vertex.

If all nontrivial syndrome edges emanating from the vertex are in the future of that vertex, we refer to the vertex as extremal. If a vertex $\mathsf{v}^*$ is extremal, then we apply a correction operator on on the face (qubit) $\mathsf{f}^*$ enclosed by the future syndrome edges connected to $\mathsf{v}^*$ such that the syndrome on all the edges of $\mathsf{f}^*$,  $\mathsf{e}^*\subset\mathsf{f}^*$ is flipped. This effectively sweeps the future syndrome edges connected to the vertex to the other two edges of the face. We apply this step for all extremal vertices in parallel. Under a sufficient number of such consecutive implementations of the sweep rule, most syndrome configurations clean up in a time or number of steps that scales linearly with the system size. 

We now discuss the sweep rule for the 3D FSC on a lattice with open boundaries. For the fractal lattice $\widetilde{\mathcal{L}}^*$, we need a modification to the sweep rule because configurations of syndromes joining the holes can get trapped if one uses the regular sweep rule. These configurations are resistant to the regular sweep rule because there are no extremal vertices even if the sweep direction is changed. This is similar to the configuration of a pair of parallel lines that span the lattice. We illustrate a trapped configuration between holes of the 3D FSC via a 2D projection. To relieve these trapped configurations connecting holes, the sweep rule needs to be modified on the hole boundary~\footnote{This differs from the technique of alternating sweep directions used to address syndromes trapped at the global boundary as addressed in Ref.~\cite{vasmer2020cellular}.}. 

For the fractal lattice $\widetilde{\mathcal{L}}^*$ considered in this work, even though some of the qubits are removed in the holes made in a regular lattice $\mathcal{L}^*$, our goal is to imitate the sweep decoder in the regular lattice as much as possible. The modification to the sweep rule is only on the vertices ${\mathsf{v}}^*_h$ that are at the boundary of the $m$-hole. Note that in the (original) fractal lattice ${\mathcal{L}}$, these correspond to the cubes ${\mathsf{c}}_h$ that form the outermost layer of cubes inside the $m$-hole as shown in Fig.~\ref{fig:FSC_and_mholes}(b). Even though there are no qubits on the hole faces ${\mathsf{f}}^*_h$, we include the vertices ${\mathsf{v}}^*_h$ in sweep indices, which is the set of vertices considered for application of the sweep rule. Now, we note that in the case of perfect syndrome measurements, for the bulk vertices, only an even number of syndrome edges in ${\mathsf{e}}^*$ can be incident on each vertex. This is not necessarily true for the vertices ${\mathsf{v}}^*_h$. If a vertex on the hole boundary ${\mathsf{v}}^*_h$ is connected to zero or more than one syndrome edge, then we apply the regular sweep rule on that vertex. If ${\mathsf{v}}^*_h$ has exactly one syndrome edge connected to it, then we first check whether this syndrome edge is in its future or not. If yes, then this vertex is already extremal; however, since there is only one syndrome edge, the sweep rule is unable to specify the correction. We check the future edges of this vertex that lie on the surface of the hole and, hence, have no associated stabilizer. Among these edges, we choose one randomly and mark it as an imaginary syndrome edge associated with the vertex. Once we have included this imaginary syndrome, the regular sweep rule can specify a correction operator for the extremal vertex ${\mathsf{v}}^*_h$. To summarize, the modification to the sweep rule is as follows: the vertices on the hole boundary are allowed in the sweep indices, and if a hole vertex has a single nontrivial syndrome edge connected to it, then one out of its future edges which sits on the hole boundary is assigned a syndrome. 

% There are certain trapped configurations for which a local cleaning strategy is not known. For instance, syndrome configurations that are composed of pairs of parallel lines stretching across the lattice, as shown in Fig.~\ref{fig:sweep_illust} are not shrinkable since there are no extremal vertices for any sweep direction. The probability of these configurations, however, decays exponentially with the system size.

For the case of imperfect measurements with an error rate $q$, we implement the modified sweep rule without change. The intuition behind this is as follows: faulty measurement could lead to broken string-like syndromes due to certain missing syndrome edges, ``ghost" (false) syndromes, or, generally speaking, local deformation of the syndromes, all with small probability $q$. This leads to missing or wrong updates in certain local regions, which is effectively converted again to the phase-flip error with a small probability in the next round. Therefore, the decoder, applied as if to a memory with pure $X$ error, is still expected to have a threshold.

In Sec.~\ref{sec:algoproof} below, we prove that there exist non-zero thresholds for the bit-flip noise with the sweep decoder. The essence of the proof is that the holes can, in a sense, be treated like errors since, in the worst case, any error that attaches to a hole may take extra time that is linear in that hole size to be corrected. However, due to the scale-invariance of the fractal, the number of holes of a given size drops off exponentially. Hence, the errors that attach themselves to the holes can be cleaned up sufficiently quickly with high probability since they will be far away from other holes of similar size. More specifically, while large holes are close to many smaller holes, their distance from a hole of a similar size is on the order of the size of the hole itself. As such, any smaller error will only, at most, feel the effect of one of the larger holes (not multiple) that it connects to. Therefore, following similar arguments for the 3D surface code, error patches of large size will be increasingly improbable with growing size, and while the correction of an error may be delayed if connected to a large hole, it will not result in connecting multiple large holes. It will eventually be cleaned up by the sweep decoder in a time that is linear in the size of the region that bounds the error and the holes connected to it. In the proof, we formalize these concepts in the language of error chunks and connected components as discussed in Refs.~\cite{BravyiHaah2013_3Dcubic, kubica_ca}.

% \section{Algorithm and Proof}
 
\subsection{Proof of sweep decoder threshold in the FSC}
% \label{app:sweepdecoder}
\label{sec:algoproof}

% \tomas{Should probably change the symbol for $L$ in this Appendix to match the main text. I will do it.} 
Let $\mathcal{E}$ be the set of all possible Pauli~$X$ errors in the 3D fractal code. By definition, the fractal will have hole sizes at each iteration of cubic length~$H_i = \alpha D^i$, where $\alpha = b/a \in (0,1)$ is related to the $a$ and $b$ terms introduced in the definition of the fractal codes~$\text{FC}(a,b)$ and $D$ is the smallest sized cube when partitioning. We define all single-qubit errors to be level-0 chunks. All subsequent level chunks are defined recursively as follows: a level-$n$ chunk is the disjoint union of two level-$(n-1)$ chunks, $E^{[n]} = E^{[n-1]}_1 \sqcup E^{[n-1]}_2$ such that the diameter of the level-$n$ chunk is smaller or equal to $\alpha D^n / 2$, that is $\text{diam}(E[n]) \le \alpha D^n/2$.

\begin{figure}[H]
\begin{center}
\includegraphics[width=0.4\textwidth]{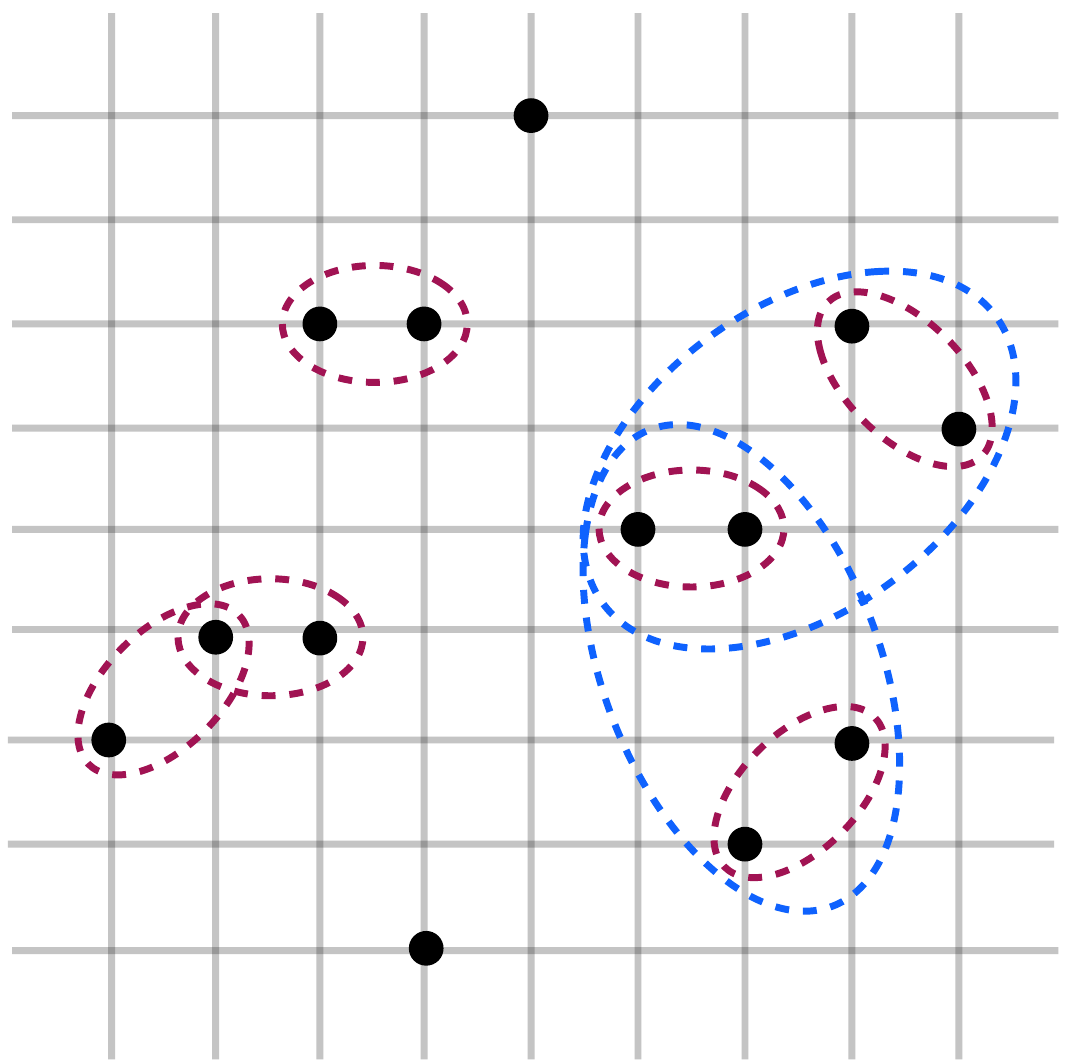}
\caption{Example of different level-$n$ chunks of errors in two dimensions. The black dots represent errors belonging to $E_0$ by definition. In this example, the level-1 chunks, dashed red lines, are pairs of errors from~$E_0$ that are distance at most two from one another (in Manhattan distance). The level-2 chunks are formed from pairs of disjoint level-1 chunks such that all errors are within distance four from one another (again, Manhattan distance). Since there are no pair of disjoint level-2 chunks, there are no level-3 chunks (and higher). All errors that are not surrounded by a dashed line will therefore be in~$F_0 = E_0 \backslash E_1$, all errors that are only surrounded by red dashed lines will be in~$F_1$, and finally, the errors surrounded by the blue dashed lines will belong to $F_2$. }
\label{fig:ErrorChunks}
\end{center}
\end{figure}

Given some error pattern $E$, we define the set of errors~$E_n$ to be the union over all level-$n$ chunks, and as such: $E = E_0 \supseteq E_1 \supseteq E_2 \supseteq \cdots \supseteq E_{m+1} = \emptyset$, where $m$ is the smallest integer such that $E_{m+1}$ is empty which will always be satisfied as the configuration space of errors~$\mathcal{E}$ is countable. Finally, by defining $F_i = E_i \backslash E_{i+1}$, every individual error from $E$ can be classified into one of the sets $F_0, \cdots, F_m$, which we call the chunk decomposition of~$E$. See Fig.~\ref{fig:ErrorChunks} for an illustration of different level-$n$ chunks. 

We would like to prove then that given an independent noise model for each of the physical qubits, the probability of generating a high-level chunk decreases exponentially with the level of the chunk, even in the presence of holes. We do so following the techniques developed in percolation theory, mirroring the procedure from Refs.~\cite{BravyiHaah2013_3Dcubic, kubica_ca}.

Consider a random error $E \in \mathcal{E}$ of independently generated single qubit Pauli~$X$ errors. Let $B_n(x)$ be a fixed box of linear size~$H_n = \alpha D^n$ centered at $x$ and $B_n^{+}(x)$ to be a box of linear size $3H_n$ centered at~$x$. We define the following probabilities:
\begin{align*}
&p_n(x) = \text{Pr} \big[ B_n(x) \text{ has a non-zero overlap with level-$n$}\text{chunk of $E$} \big] \\
&\tilde{p}_n(x) = \text{Pr} \left[ B_n^{+}(x) \text{ contains a level-$n$ chunk of $E$} \right]
\end{align*}
\begin{align*}
&q_n(x) = \text{Pr} \big[ B_n^{+}(x) \text{ contains 2 disjoint level-$(n-1)$}\text{chunks of $E$} \big] \\
&r_n(x) = \text{Pr} \left[ B_n^{+}(x) \text{ contains a level-$(n-1)$ chunk of $E$} \right]
\end{align*}

For a given~$x$, if $B_n(x)$ has a non-zero overlap with a level-$n$ chunk since the diameter of a given level-$n$ chunk is at most~$\alpha D^n/2$, that chunk will necessarily be contained within $B_n^{+}(x)$. Moreover, if $B_n^{+}(x)$ contains a level-$n$ chunk then by definition it must contain 2 disjoint level-$(n-1)$ chunks. As such, $p_n(x) \le \tilde{p}_n(x) \le q_n(x)$. Moreover, since an event that contributes to $r_n(x)$ is the result of two independent events contributing to $q_n(x)$, by the van~den~Berg and Kesten inequality, we must have $q_n(x) \le r_n(x)^2$.

Since the box $B_n^{+}(x)$ is the disjoint union of $(3D)^3$ boxes of linear size~$H_{n-1}$, we have the following:
\begin{align*}
r_n(x) \le \sum_{x'} p_{n-1} (x') \le (3D)^3 \max_{x'} p_{n-1}(x')
\end{align*}
and thus,
\begin{align*}
p_n(x) \le r_n(x)^2 \le (3D)^6 \max_{x'} p_{n-1}^2(x').
\end{align*}

Finally, if we iterate the process over all $n$:
\begin{align*}
p_n(x) \le (3D)^{-6} \left( (3D)^6 p_0 \right)^{2^n},
\end{align*}
and $p_0$ is just the bare single qubit error rate~$\epsilon$.
Thus, the probability of a given chunk is exponentially suppressed in its size. The final piece that is needed before arguing for the ability of the sweep decoder to succeed with high probability is the notion of a connected component. A set of errors $E \in \mathcal{E}$ is a $R$-connected component if it cannot be split into two disjoint non-empty sets $E = M_1 \sqcup M_2$ such that the distance $d(M_1, M_2) > R$. The critical fact is that for a given error~$E$, any set of errors belonging to~$F_n$ will necessarily belong to a $(H_n)$-connected component that is sufficiently separated away from other elements of~$E_n$. The following Lemma summarizes this.

\begin{lemma}[\cite{BravyiHaah2013_3Dcubic}]
\label{lem:connectedcomponent}
Let $H_n = \alpha D^n$, such that $D \ge 6$ and let $M \in \mathcal{E} $ be a set of errors belonging to a $H_n$-connected component of $F_n$. Then, $M$ has a diameter $\le H_n$ and $d(M, E_n \backslash M) > H_{n+1}/3$.
\end{lemma}

\begin{proof}
The claim is that for any pair of errors $m \in M \in F_n = E_n \backslash E_{n+1}$ and $p \in E_n$ we have the following: $d(m,p) \le H_n$ or $d(m,p) > H_{n+1}/3$. We prove this by contradiction, that is suppose $H_n < d(m,p) \le H_{n+1}/3$. Let $M_n, P_n$ be the level-$n$ chunk to which $m$ and $p$ belong to, respectively. Therefore, by definition, since their respective diameters are bounded by~$H_n/2$ and $H_n < d(m,p)$, they must be disjoint. However, we can also bound the overall distance between $M_n$ and $P_n$: 
\begin{align*}
d(M_n, P_n) &\le \text{diam}(M_n) + d(m,p) + \text{diam}(P_n) \\
& \le H_n/2 + H_{n+1}/3 + H_n/2 \\
&= \alpha D^n + \alpha D^{n+1}/3 \\
& \le \alpha D^{n+1}/6 + \alpha D^{n+1}/3 = H_n/2,
\end{align*}
as such the union of $M_n$ and $P_n$ form a level-$(n+1)$ chunk, which is a contradiction to $M_n \subseteq M \in F_n$.
\end{proof}

\begin{figure*}[htb]
\centering
\sidesubfloat[]{
         \includegraphics[width=0.35\textwidth, trim={10cm 6cm 10cm 6cm},clip]{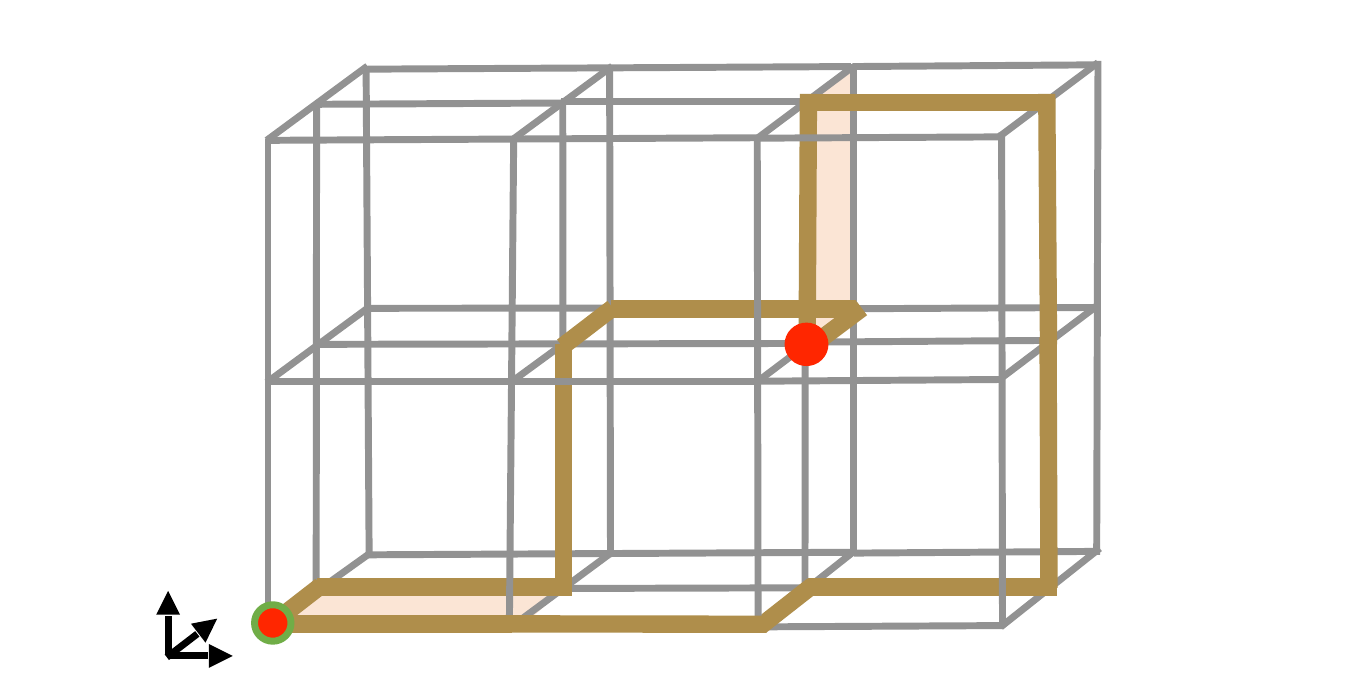}}
\sidesubfloat[]{
         \includegraphics[width=0.35\textwidth, trim={10cm 6cm 10cm 6cm},clip]{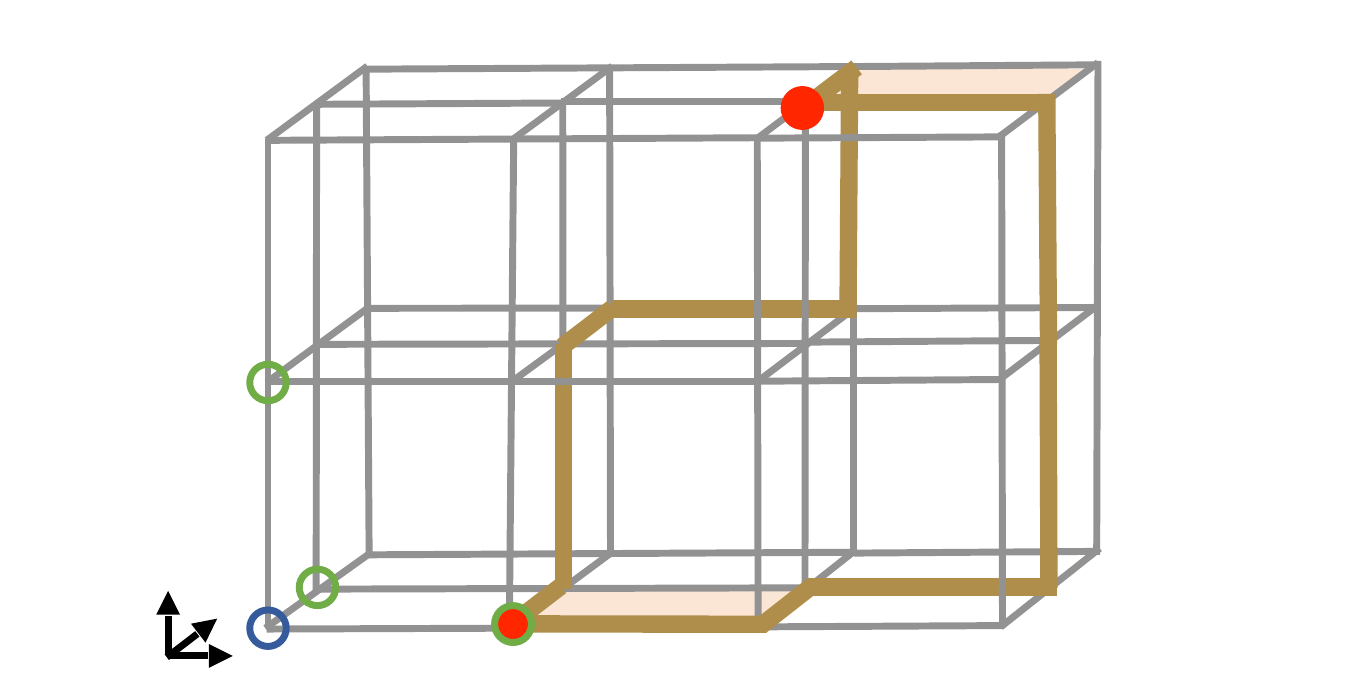}
}\\
\sidesubfloat[]{
         \includegraphics[width=0.35\textwidth, trim={10cm 6cm 10cm 6cm},clip]{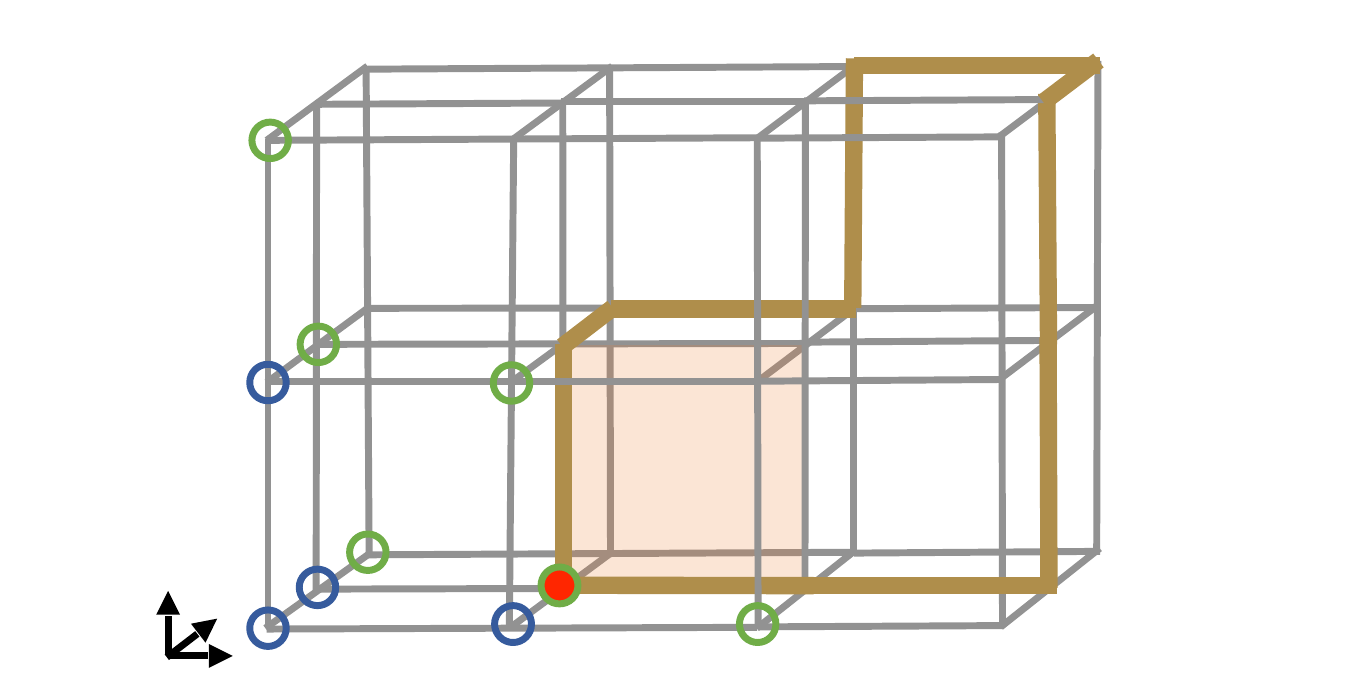}}
         \vspace{4mm}
\sidesubfloat[]{         
         \includegraphics[width=0.35\textwidth, trim={10cm 6cm 10cm 6cm},clip]{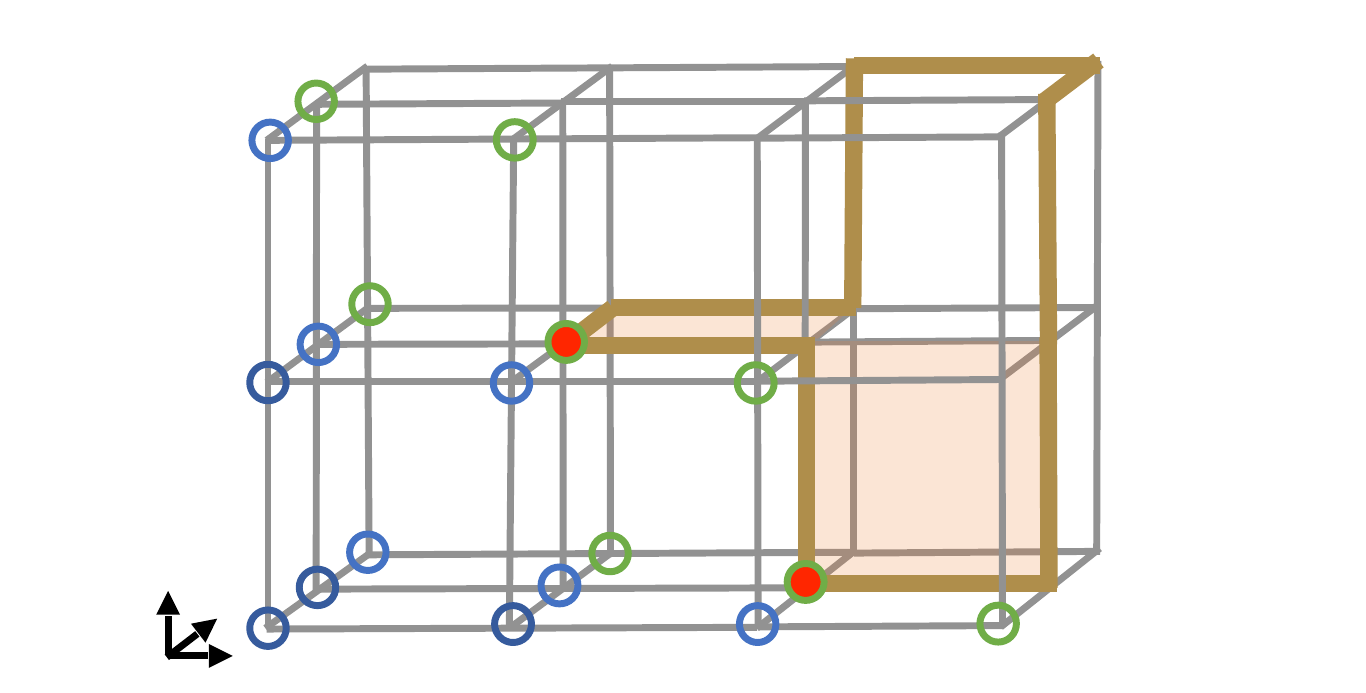}}\\
\sidesubfloat[]{
         \includegraphics[width=0.35\textwidth, trim={10cm 6cm 10cm 6cm},clip]{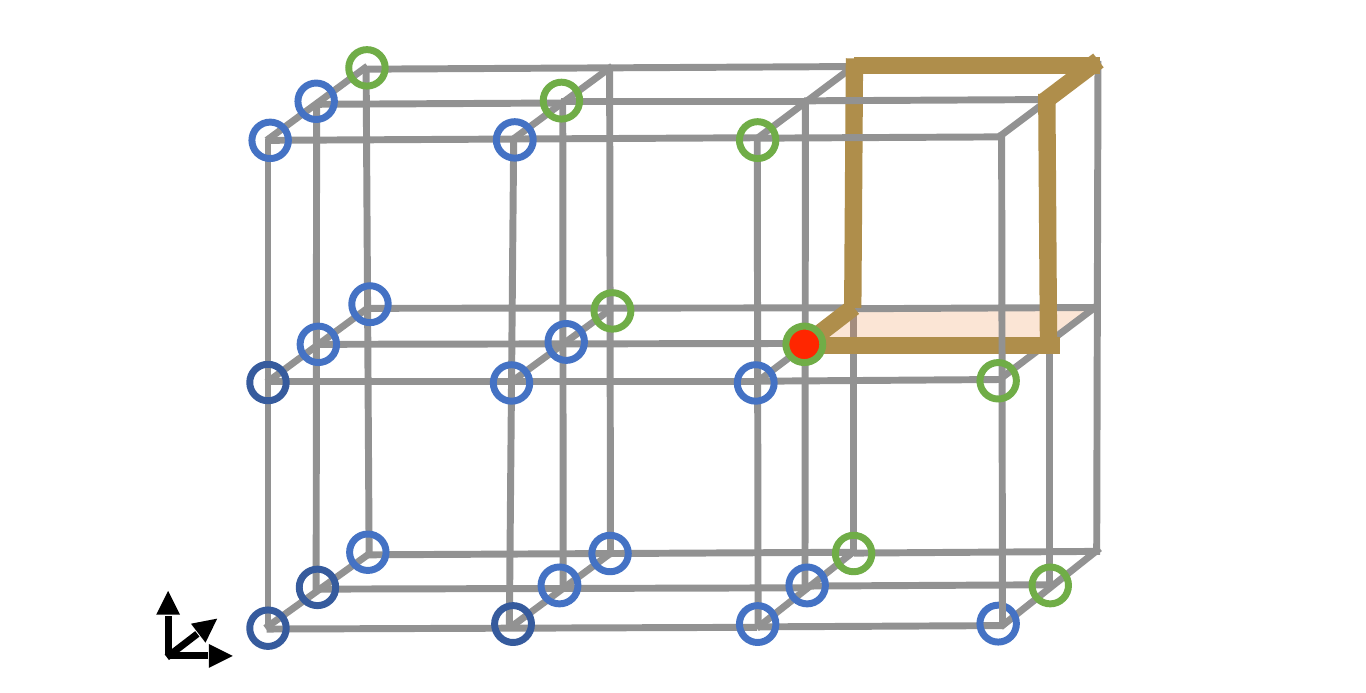}}
\sidesubfloat[]{         
         \includegraphics[width=0.35\textwidth, trim={10cm 6cm 10cm 6cm},clip]{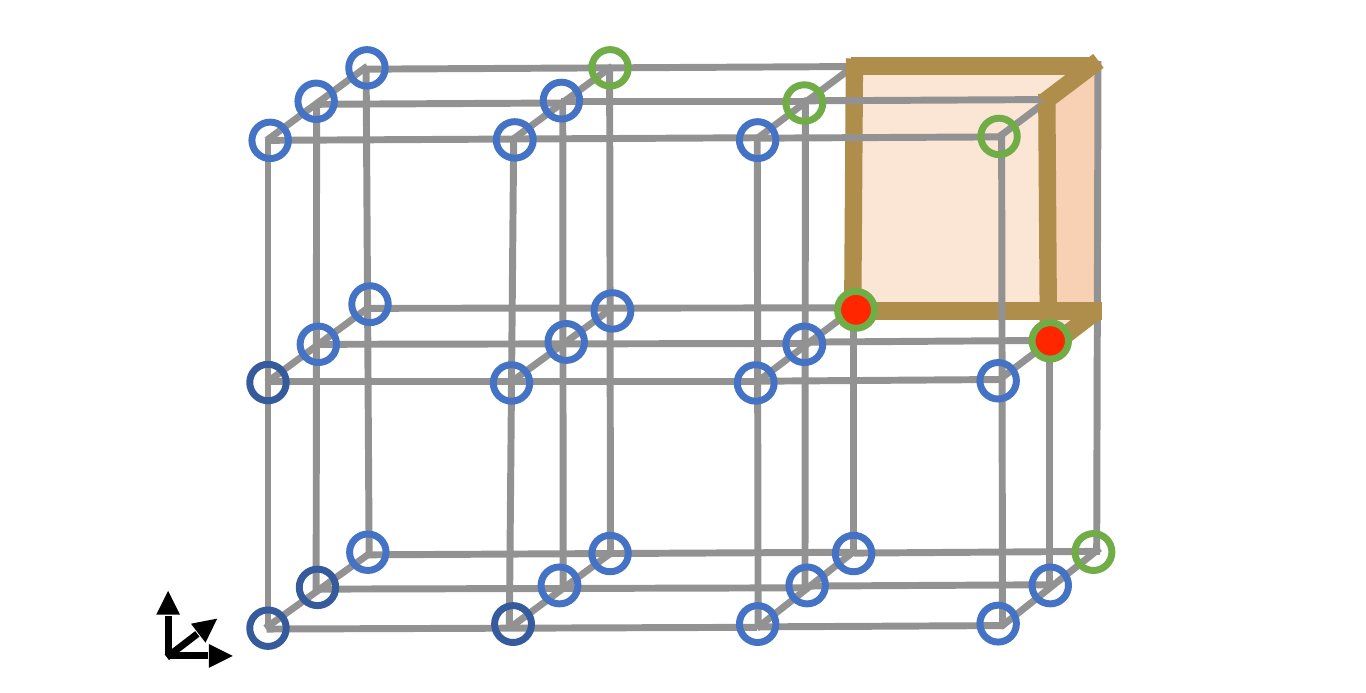}}
\caption{Example of an error sweeping. Initial error configuration is given in~(a) with sweep direction in the (+1,+1,+1) direction. At each step, we are guaranteed to sweep away from the green circles if an error is present there. Green circles represent a plane moving through space such that the syndrome is guaranteed to no longer be in the past light cone of this plane relative to the sweep direction. Blue circles represent vertices that are in the past of the plane. Extremal vertices are given by circles filled in red; notice that whenever the syndrome is present at a green circle, the corresponding vertex is extremal.} 
\label{fig:SweepEx}
    
\end{figure*}

The conclusion of this Lemma is to separate out all errors into their various connected components and to show that the decoder will successfully address each connected component separately as they are sufficiently far away from one another. The complication arises due to the presence of holes, which can delay the cleaning up of a component, namely if an $H_n$-connected component of errors (thus of size $\mathcal{O}(D^n)$) was connected to a hole of size $\mathcal{O}(D^p)$, where $p>n$, then the resulting correction would take time~$\mathcal{O}(D^p)$ (that is linear with the larger hole size). We show below that our variant of the sweep decoder can successively correct for larger and larger errors, even in the presence of holes.

\begin{figure*}[htb]
     \centering
\sidesubfloat[]{
         \includegraphics[width=0.3\textwidth, trim={10cm 5cm 10cm 5cm},clip]{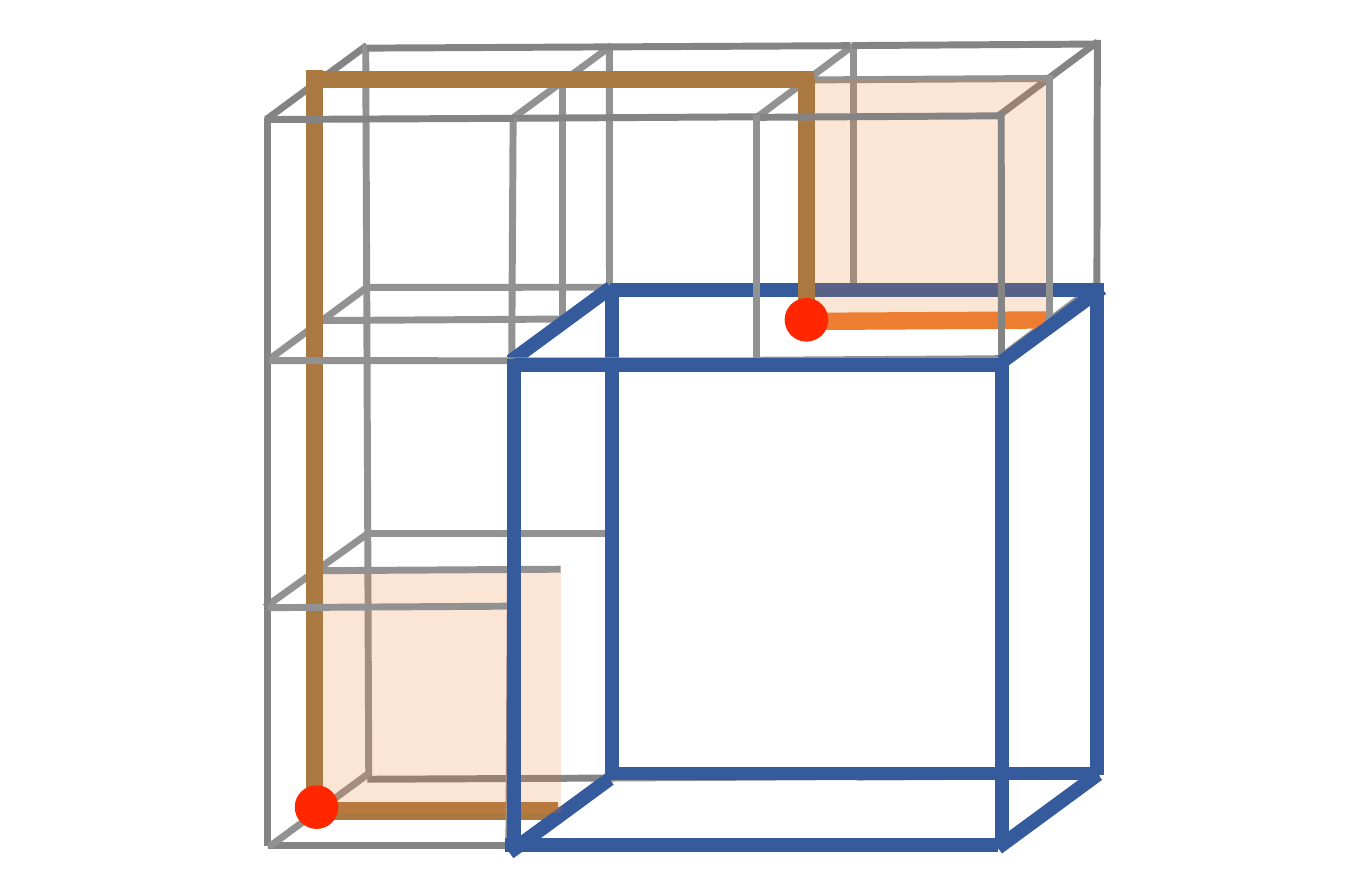}}
\sidesubfloat[]{
         \includegraphics[width=0.3\textwidth, trim={10cm 5cm 10cm 5cm},clip]{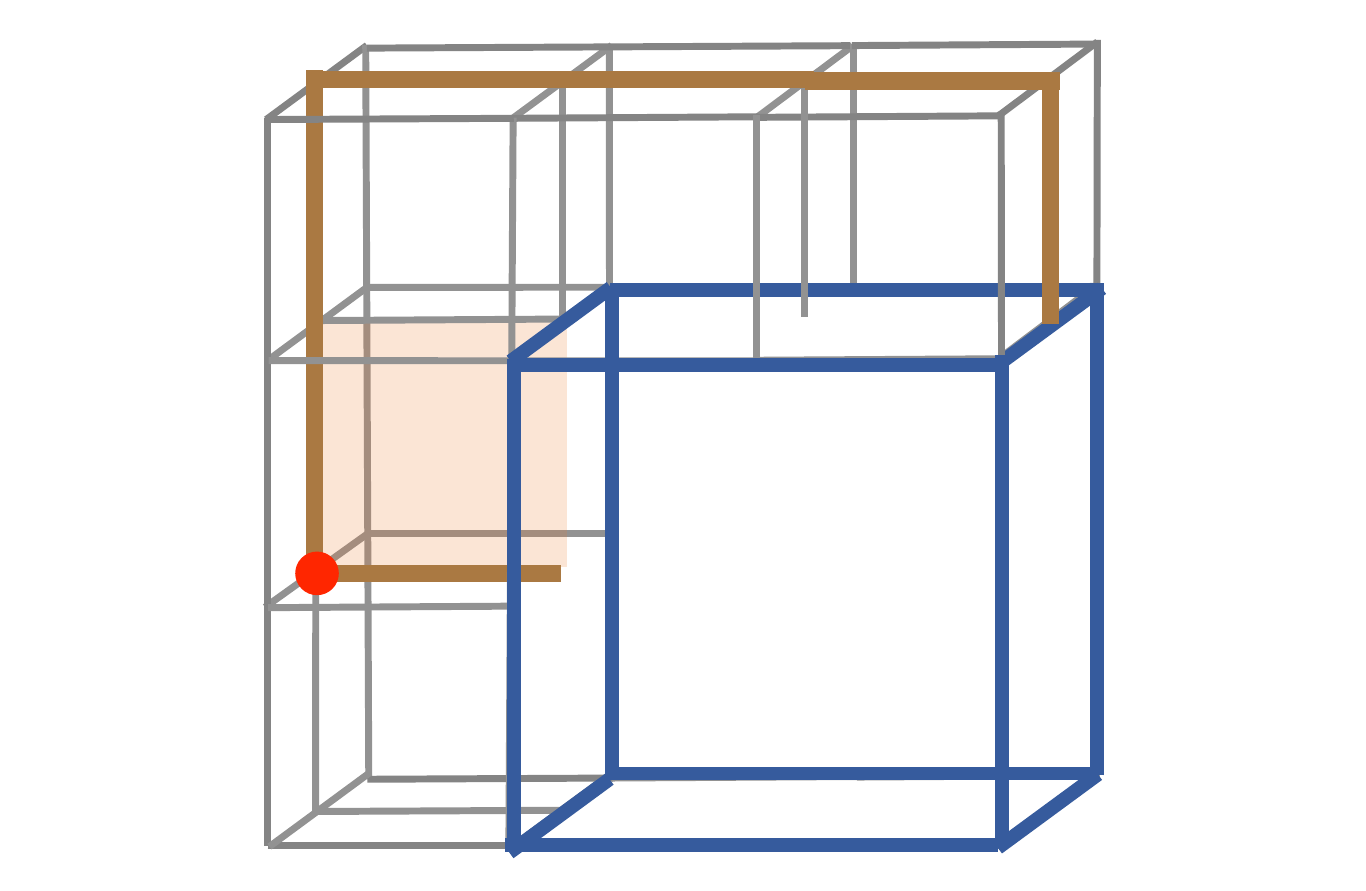}}\\
\sidesubfloat[]{
         \includegraphics[width=0.3\textwidth, trim={10cm 5cm 10cm 5cm},clip]{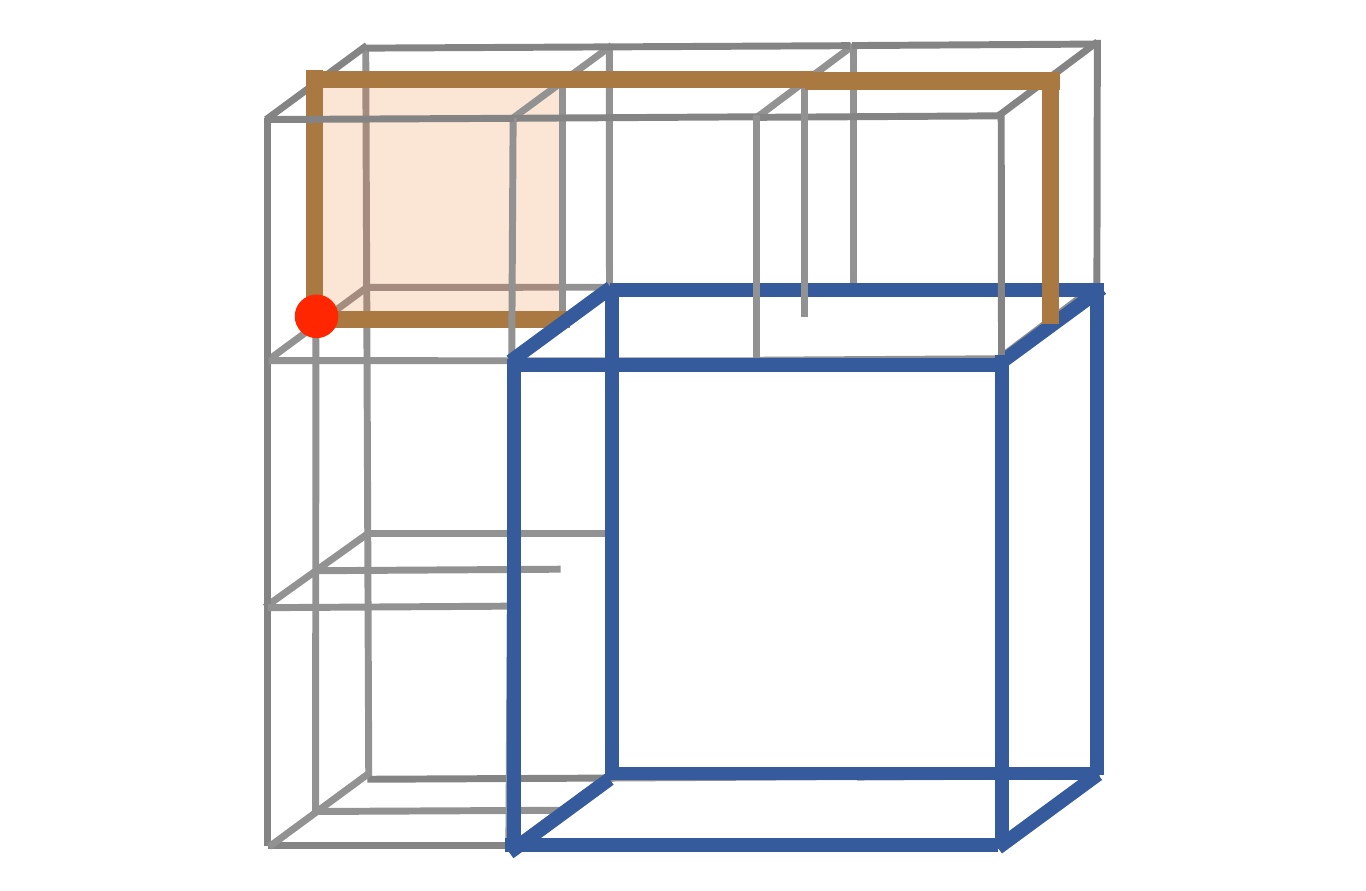}}
\sidesubfloat[]{
         \includegraphics[width=0.3\textwidth, trim={10cm 5cm 10cm 5cm},clip]{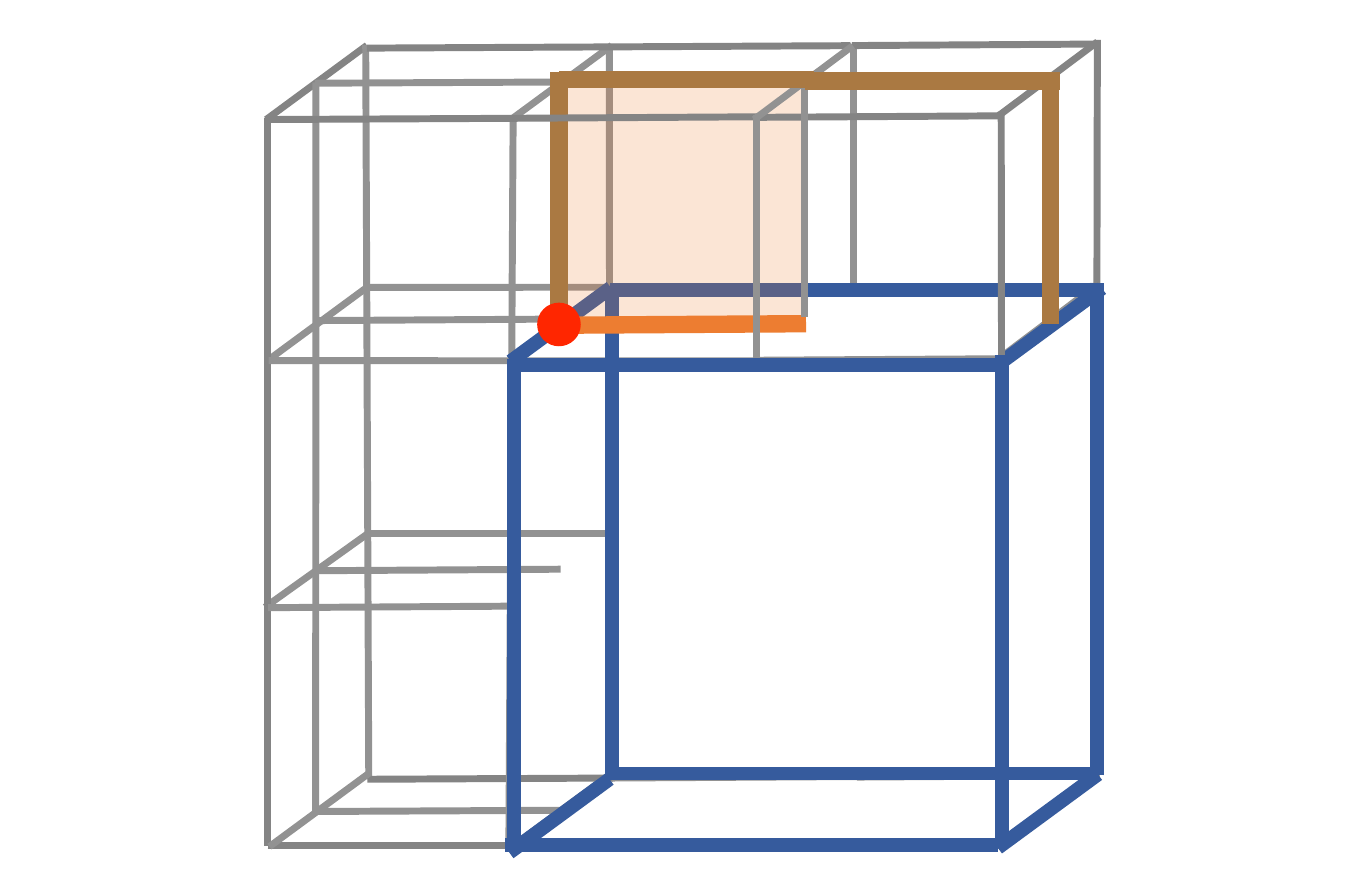}}
\caption{Example of an error sweeping connected to a hole. Initial error syndrome given in brown in~(a) with sweep direction in the (+1,+1,+1) direction. The hole is shown by the blue edges. Any sweep correction cannot take one outside the set of edges given since they are all within a cuboid surrounding the error syndrome and hole. The extremal vertices are shown by red points at each step, imaginary syndromes are introduced along the surface of the hole in orange to make the appropriate vertices extremal as described in the modified sweep decoder section.}
\label{fig:SweepHole}
\end{figure*}

In order to simplify the discussion of the sweep correction of errors, we provide a few useful definitions. Suppose we have an error $E \in \mathcal{E}$, we define the \textit{error envelope} $V(E)$ of~$E$ to be the smallest cuboid that encases~$E$. This definition will be useful for upper-bounding the sweep time of errors. We begin by presenting the following result for the sweeping of an individual error membrane without the presence of holes; we also give an illustration of a sweeping in Fig~\ref{fig:SweepEx}.

\begin{lemma}
Consider an error membrane $E \in \mathcal{E}$ with error envelope $V(E)$ whose linear dimensions are: $l_x,\ l_y, \ l_z$. The sweep decoder will correct for such an error in at most $(l_x + l_y +l_z - 1)$ sweep steps. 
\end{lemma}

\begin{proof}
Without loss of generality, suppose the error envelope $V(E)$ is formed by the cube whose linear coordinates are given by $[0, l_i]$, where $i \in \{x,y,z\}$ and the sweep direction follows the vector~$(+1,+1,+1)$. Then, in the first step, the boundary of the error will be swept away from the corner~$(0,0,0)$, and thus the error can no longer touch this point. In the second step, the error will be swept away from the three points distance~1 away from the origin in Manhattan distance, that is $(1,0,0),\ (0,1,0), \ (0,0,1)$. In an iterative manner, in the $n$-th time step, the error will be necessarily swept away from all points that are distance $(n-1)$ in from the origin. Since the furthest point from the origin is at the opposite corner of the cuboid $(l_x,l_y,l_z)$ and any error containing that point must have at least one point that distance~2 closer to the origin, we are guaranteed to clean up the error after $(l_x+l_y+l_z -1)$ time steps. 
\end{proof}

In the presence of holes, errors may take longer to be corrected; however, we can also bound the number of time steps an individual error membrane will take to be corrected. This is pictorially shown in Fig.~\ref{fig:SweepHole}, where the error envelope must also surround the hole.

\begin{corollary}
Given an error membrane $E \in \mathcal{E}$ that is connected to holes $H_1, \ \cdots, H_j$. Consider the error envelope $V(E\cup H_1 \cup \cdots H_j)$ that encompasses the error and all holes whose linear dimensions are $l_x,\ l_y, \ l_z$. Then, the sweep decoder will correct for such an error in at most $(l_x + l_y +l_z - 1)$ sweep steps.
\end{corollary}

\begin{proof}
The proof of the previous Lemma can be generalized to account for the case where the error membrane is connected to multiple holes. As stated, if we choose an error envelope that encompasses all of the holes connected to the membrane (as well as the membrane itself), then the resulting correction will never leave the envelope. Thus, when sweeping from one of the corners in the sweep direction, we are guaranteed to clean the error in the number of sweep steps given.
\end{proof}

The upshot of the above Lemma and Corollary is that a connected component of radius~$R$ will be corrected in time~$\mathcal{O}(R)$ if not connected to any holes, while it will be corrected in time at most~$\mathcal{O}(R+R_h)$ if connected to a hole of radius~$R_h$. If connected to multiple holes, the corresponding correction time will be linear in the radius of the largest hole or the error itself, whichever is larger.

Let $E \in \mathcal{E}$ be an error instance. Consider all $H_0$-connected components of $F_0$ that are distance greater than $H_0$ from any hole, denoted by $Q_0$. Given the diameter of these components is at most $H_0$, we know the decoder will correct them in time linear to the diameter, assuming there is no interference from other errors or holes. Given a connected component~$Q_0$ is the distance at least $H_0$ away from any hole and distance $H_{1}/3$ away from any other error, this connected component will successfully be corrected. As such, all errors $F_0$ will be successfully corrected unless they were within distance~$H_0$ from any hole; we label the remaining errors as $T_0 = F_0 \backslash Q_0$.

Consider now the elements from $F_1$. Again, we break them into two classes, the elements $Q_1$ which are in $H_1$-connected components whose distance is greater than $H_1$ from any hole of size $H_{k>1}$, and the errors $T_1$ which are~$F_1 \backslash Q_1$. Any error in $Q_1$ will be corrected in time linear in~$H_1$ as these elements are either independently corrected or they are affected by either the uncorrected errors~$T_0$ or holes of linear size~$H_1$. However, since both these holes and errors are of size at most~$H_1$, they will only affect the cleanup time and size of the correction bubble of~$Q_1$ by a constant factor in~$H_1$. The critical point here is that any error from $Q_1$ can only see their correction time increase by~$\mathcal{O}(H_1)$ due to the uncorrected errors from $F_0$ or holes of size~$H_1$. Since neither of these objects is larger than the element from $Q_1$, which is also of size~$\mathcal{O}(H_1)$, such objects will not affect the decoder's ability to clean up an element from~$Q_1$.

We iterate this process for any $n$. Consider elements from $F_n$, breaking them into two classes: the elements labeled $Q_n$, which are $H_n$-connected components whose distance is greater than $H_n$ from any hole of larger size $H_{k > n}$ and the complement of such elements $T_n = F_n \backslash Q_n$. Any element of $Q_n$ will be cleaned up in time linear in $H_n$ as it will either attach itself to a hole of similar size (or smaller) or will be affected by smaller uncorrected errors~$T_{j<n}$, in either case this will not change the cleanup time for such an error, and it will be corrected. Thus by induction any level-$n$ error that is smaller than the system size will be corrected in time at most linear in the size of the largest chunk. Therefore, to summarize, errors from $Q_n$ are corrected in time~$\mathcal{O}(H_n) = \mathcal{O}(D^n)$, these errors are either connected components of this size or connected to a hole of the given size.

\subsection{Sweep decoder simulation algorithm}
\label{app:sweepdecoder_algorithm}
We now list the steps of the numerical algorithm for simulating the sweep decoder as follows-
\begin{enumerate}
    \item  $N-1$ rounds of a) generating data errors $\vec{E}$ with probability $p_X$, b) performing syndrome measurements with error probability $q$, and c) applying sweep rule (modified version in case of FSC) $x$ times. We set $x=1$ in our simulations. We change the sweep direction after every $y$ rounds on the lattice with boundaries. We set $y=\log L$ where $L$ is the linear size, measured in terms of the number of cubes in the original cubic lattice (vertices in the dual lattice) along one dimension. Each implementation of the sweep rule updates the data error $\vec{E}$ to the data error times correction operator.
    \item After the above $N-1$ rounds, generate data errors again with probability $p_X$ and perform syndrome measurements. Assume perfect syndrome measurements for this last round. 
    \item Timeout session: sweep rule is implemented for $T$ steps, and sweep direction is changed after every $t$ steps. We set $T=32 L$ and $t=L$. 
    \item If the syndrome is not cleaned or if the product of total error and correction acts nontrivially on the logical subspace, the decoder fails.  
\end{enumerate}

\subsection{Sweep Decoder performance for FSC}
We implemented the sweep decoder for the FSC on the fractal lattices $\widetilde{\mathcal{L}}^*$ with boundaries. We study the performance of the sweep decoder for an error model with both phase-flip and measurement errors. We consider the measurement error rate $q$ to be the same as the physical error rate $p_X$,  \textit{i.e.}, $q=p_X$. The detailed algorithm is presented in Methods. 

We show the performance of the sweep decoder for $N=1025$ rounds in Fig.~\ref{fig:Sweep_plots} (a)-(b) and for $N=1,33$ rounds in Sec. III in Supplementary. We also summarize the results for the thresholds of the sweep decoder for different levels and number of rounds of stabilizer measurements, $N$ in Fig.~\ref{fig:Sweep_plots} (c). The last round is assumed to have perfect measurements~\footnote{This is a standard assumption in quantum error correction studies, which is valid at the readout stage where measurement errors become equivalent to data qubit errors.} while $N-1$ rounds before the last round involve noisy measurements. $N=1$ means only one round of stabilizer measurements, and those are perfect. We tabulate the numerical values of the obtained sweep decoder thresholds below,

\begin{center}
        \begin{tabular}{c|c|c|c}
     $\ell$ & 0 & 1 & 2\\
    \specialrule{.2em}{.2em}{.2em}
     $N=1$ & 15.625(8)\% & 15.59(2)\% & 15.57(2)\% \\ \hline
     $N=33$ & 2.400(1)\% & 2.471(3)\% & 2.455(2)\% \\
     \hline 
     $N=1025$ & 1.727(3)\% & 1.7331(7)\% & 1.7262(7)\% \\
    %  \hline 
    %  $N\rightarrow\infty$ & 1.7\% & \% & \% \\ 
     \hline
    \end{tabular}
\end{center}

\begin{figure}[H]
    \centering
  \sidesubfloat[]{ \includegraphics[width=0.7\columnwidth]{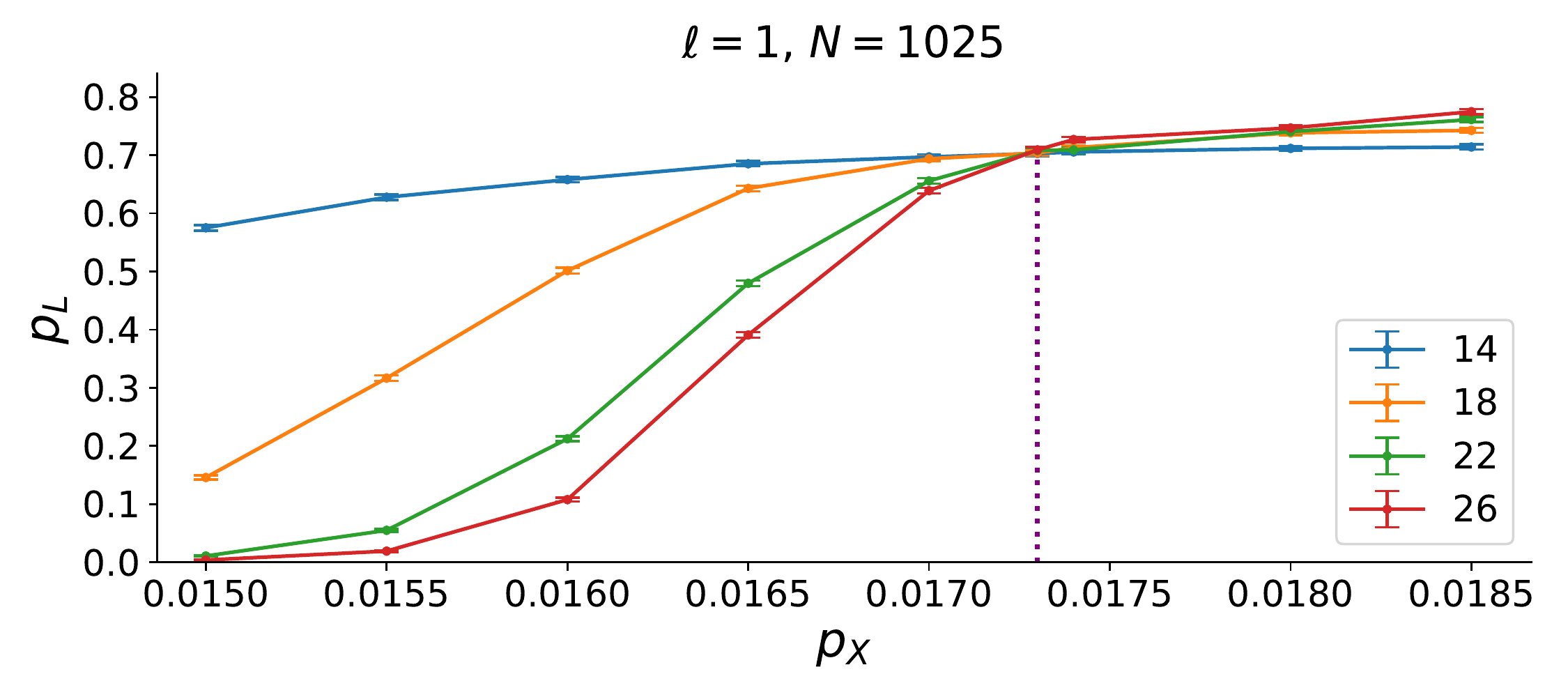}}\\
  \sidesubfloat[]{ \includegraphics[width=0.7\columnwidth]{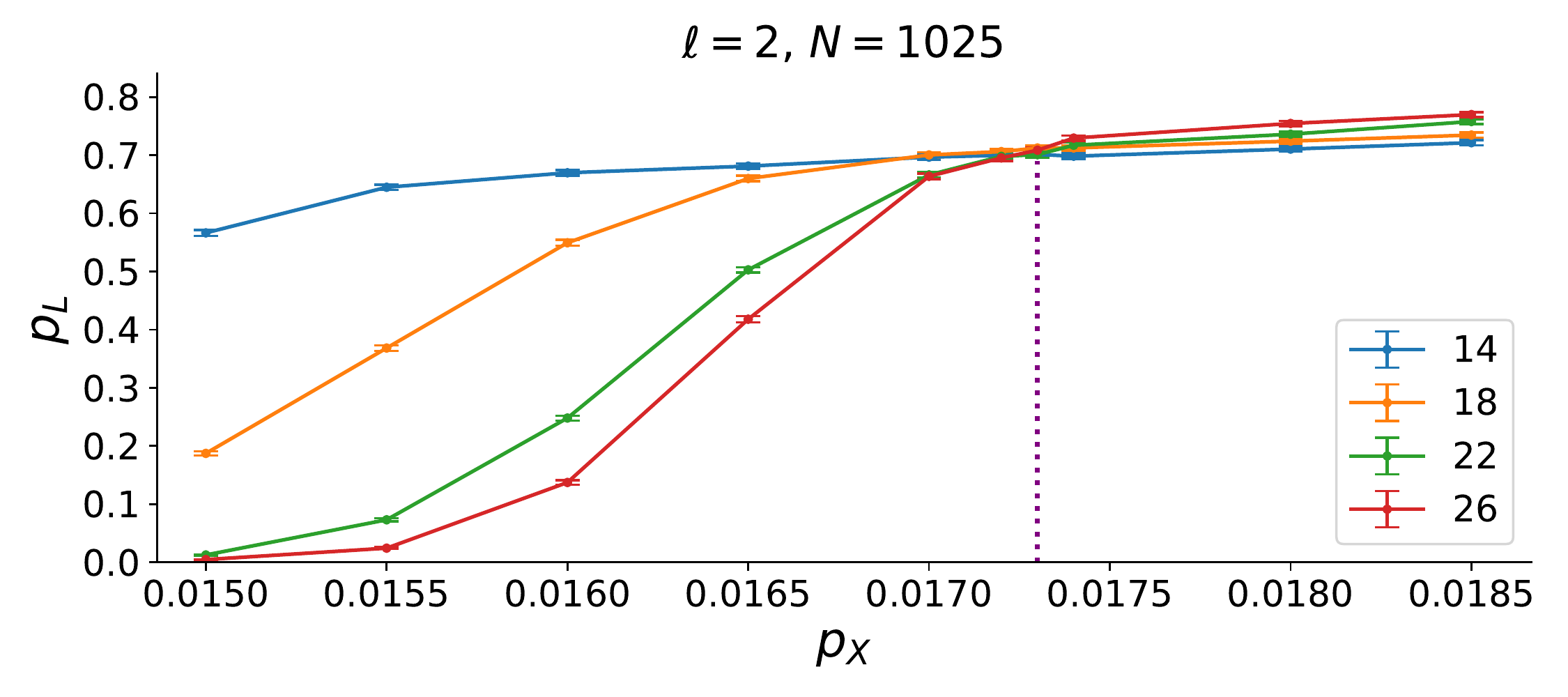}}\\
    \sidesubfloat[]{ \includegraphics[width=0.7\columnwidth]{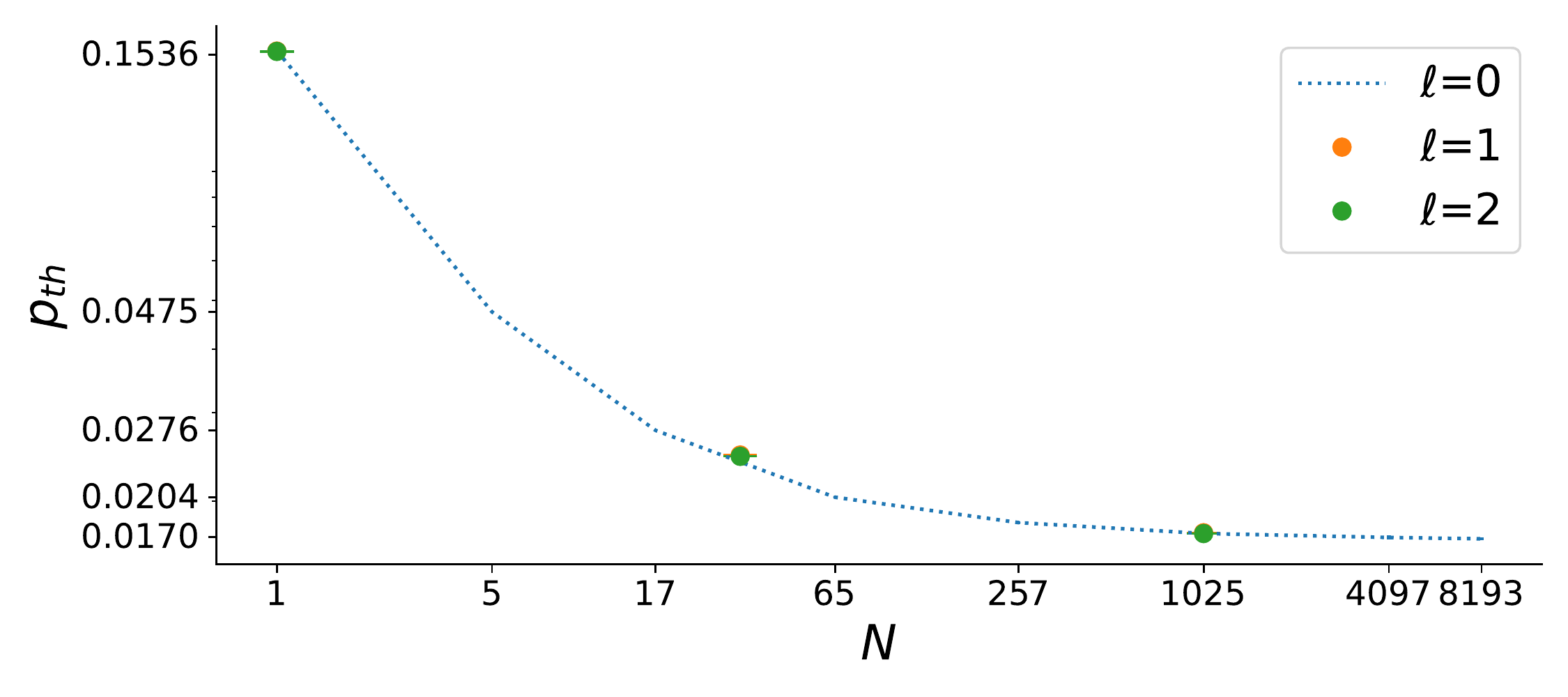}}
    \caption{(a)-(b) Sweep decoder performance for $N=1025$ rounds of stabilizer measurements for levels $\ell=1,2$ of the FSC. Logical failure rate $p_L$ for different system sizes (legends) is plotted as a function of the physical $X$ error rate $p_X=q$ where $q$ is the measurement error rate. Each data point in the simulations was obtained with 10,000 Monte Carlo runs. (c) Sweep decoder threshold $p_\text{th}$ vs number of rounds of stabilizer measurements $N$. The threshold values were obtained using the critical exponents method~\cite{wang2003confinement}. Thresholds for $\ell=0$ (3D surface code) are taken from Ref.~\cite{vasmer2020cellular}. Our results for points $\ell=0$, $N=1,33,1025$ were consistent with the curve.}
    \label{fig:Sweep_plots}
\end{figure}

The threshold obtained in the limit $N\rightarrow \infty$ is the sustainable threshold. Since the $N=1025$ thresholds for $\ell=0,1,2$ are all around $1.73\%$, we expect the sustainable threshold for both levels $\ell=1,2$ to be 1.7\%, which is the result obtained in Ref.~\cite{vasmer2020cellular} for level $\ell=0$. Moreover, due to this result, we expect this sustainable threshold value to be independent of $\ell$; hence, the sustainable threshold for the asymptotic fractal cube geometry $FC(3,1)$ is expected to be 1.7\%.

Since the $X$-distance of the code scales with the Hausdorff dimension, $D_H$ of the fractal lattice, as $L^{D_H}$, we expect the subthreshold failure rates to increase for pure (or biased) Pauli $Z$ noise as $D_H$ is lowered. In our simulations, we only looked at $D_H$ close to 3 due to the limitations on scaling up the system sizes. Hence, such differences in the subthreshold failure rates are not noticeable. 

On the other hand, the $Z$-distance of the code is unaffected under the punching of the holes; however, the effective $Z$-distance is increased, as there are fewer paths for the point-like syndrome after the smooth boundary holes are punched in. This is expected to have a positive effect on the threshold and the subthreshold failure rate for pure (or biased) Pauli $X$-noise. Below, we discuss the performance of the MWPM decoder under Pauli $X$-noise and show an improvement in the threshold error rates with increasing $D_H$. For a substantial improvement in the subthreshold failure rates, we would need to go to higher $D_H$.   

\section{MWPM decoder performance}We now discuss the thresholds of the MWPM decoder used to decode the $Z$ errors. We focus here on the case without measurement errors, which corresponds to the code-capacity error threshold. We prove the existence of the fault-tolerant threshold for the FSC using the MWPM decoder in Sec. I of Supplementary. 

As mentioned, the FSC is obtained by punching $m$-holes in the 3D surface code. For the MWPM decoder, the input decoding graph $\widetilde{\mathcal{L}}$, whose each node corresponds to a vertex stabilizer, is a subgraph of the decoding graph ${\mathcal{L}}$ of the 3D surface code, \textit{i.e.,} $\widetilde{\mathcal{L}}\subset\mathcal{L}$ since a subset of stabilizers is removed due to the $m$-holes. The main difference from the matching of the 3D surface code is that the logical string $\overline{Z}$ and error chain $E_1$ need to circumvent all the $m$-holes. 

We show the performance of the MWPM decoder, measured in terms of logical failure rate, for levels $\ell=0,1,2$ of the FSC in Fig.~\ref{fig:MWPM_lf_plots}.
We apply the critical exponents method~\cite{wang2003confinement} to the logical failure rate data to obtain the code-capacity thresholds of the MWPM decoder for each level $\ell$ of the FSC defined on $FC(3,1)$ with Hausdorff dimension $D_H=2.966$ as follows,\\
\begin{center}
    \begin{tabular}{c|c|c|c}
     $\ell$  & 0 & 1 & 2\\
    \specialrule{.2em}{.2em}{.2em}
     MWPM threshold & 2.886(4)\% & 2.931(4)\% & 2.947(5)\%  %\hline
    %  $\nu$ & 1.0308 & 1.2528 & 1.4205 
    \end{tabular}
\end{center}

Interestingly, the threshold $p_\text{th}$ increases in the case of the FSC compared to the 3D surface code case ($\ell$$=$$0$). Based on the trend at $\ell$$=$$0,1,2$, we expect the threshold to be lower bounded by $2.95\%$. This increase, relative to $\ell=0$ is due to the fact that one can upper bound the logical error rate to be proportional to the number of self-avoiding non-contractible cycles $N_\text{SAP}$ (see Sec. I in Supplementary). The input decoding graph for the FSC, $\widetilde{\mathcal{L}}$ has strictly lower $N_\text{SAP}$ than that for the 3D surface code, $\mathcal{L}$ since  $\widetilde{\mathcal{L}}\subset\mathcal{L}$. The threshold of FSC can hence be proven to be strictly higher than the 3D surface code. Moreover, the fractal decoding graph is asymptotically approaching a 2D graph in the limit of $D_H=2+\epsilon$. The code-capacity threshold of FSC is approaching that of the 2D surface code, \textit{i.e.}, $\lim_{\epsilon \rightarrow 0} p^{(2+\epsilon)}_\text{th} $$\approx$$10.31\%$ and the phenomenological fault-tolerant threshold approaches $2.9\%$  \cite{Dennis_2002} (see Sec. I in Supplementary).
\vspace{5mm}
\begin{figure}[H]
    \centering
 \includegraphics[width=0.7\columnwidth]{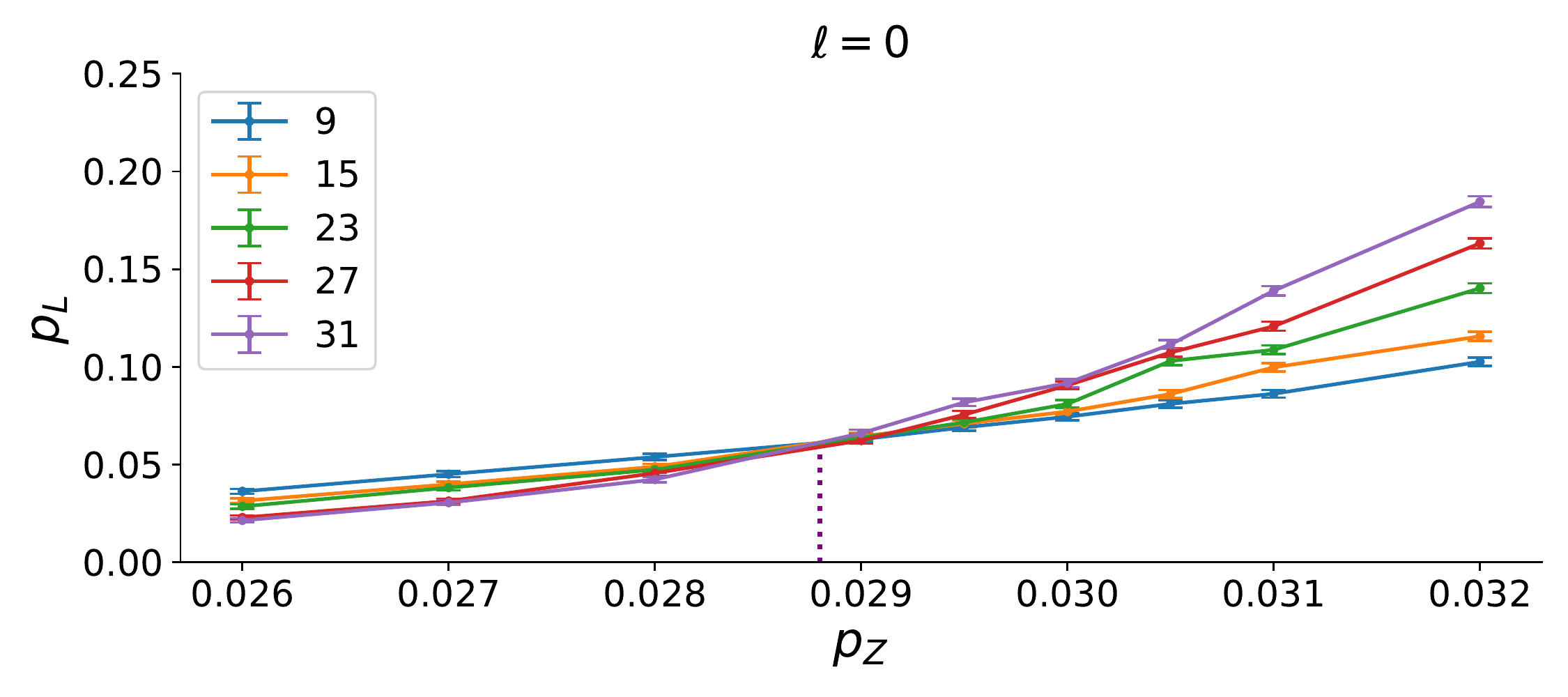}\\
\includegraphics[width=0.7\columnwidth]{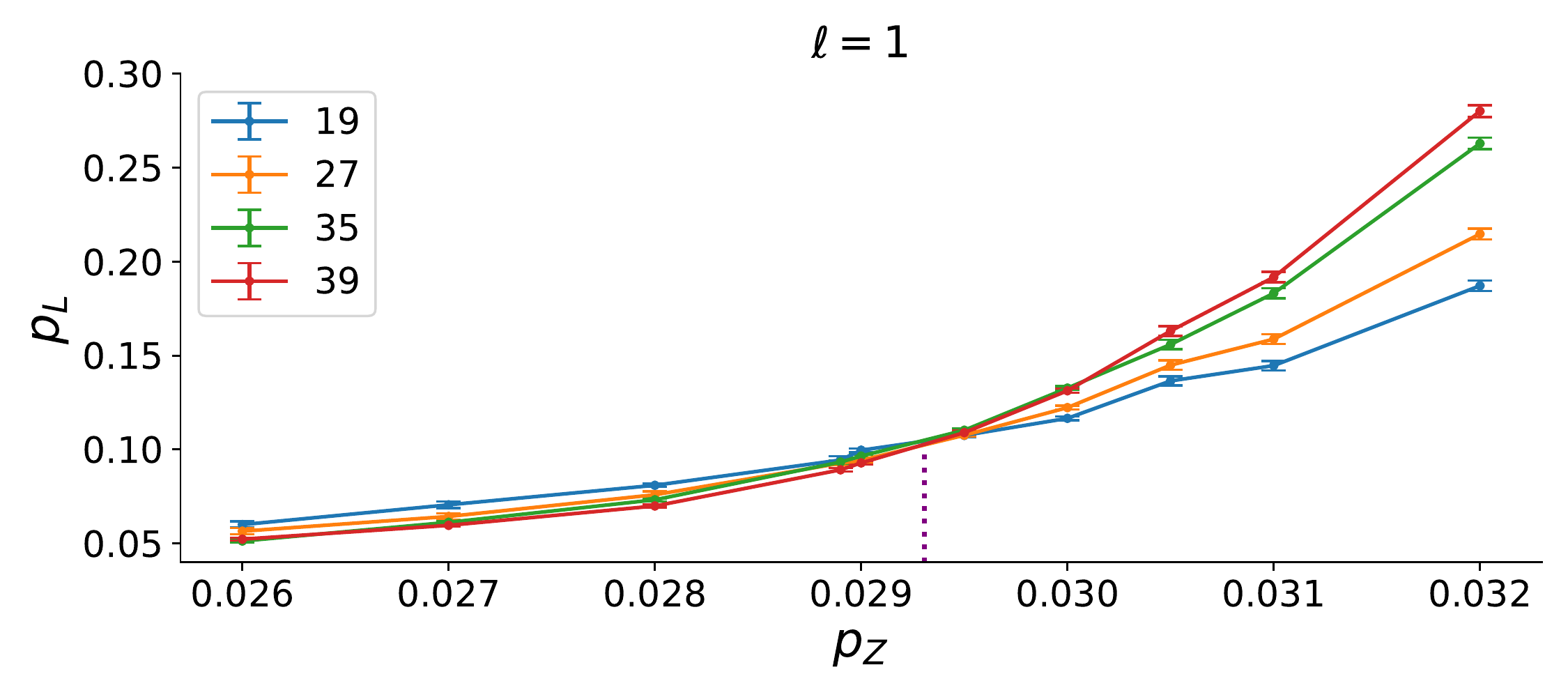}\\
 \includegraphics[width=0.7\columnwidth]{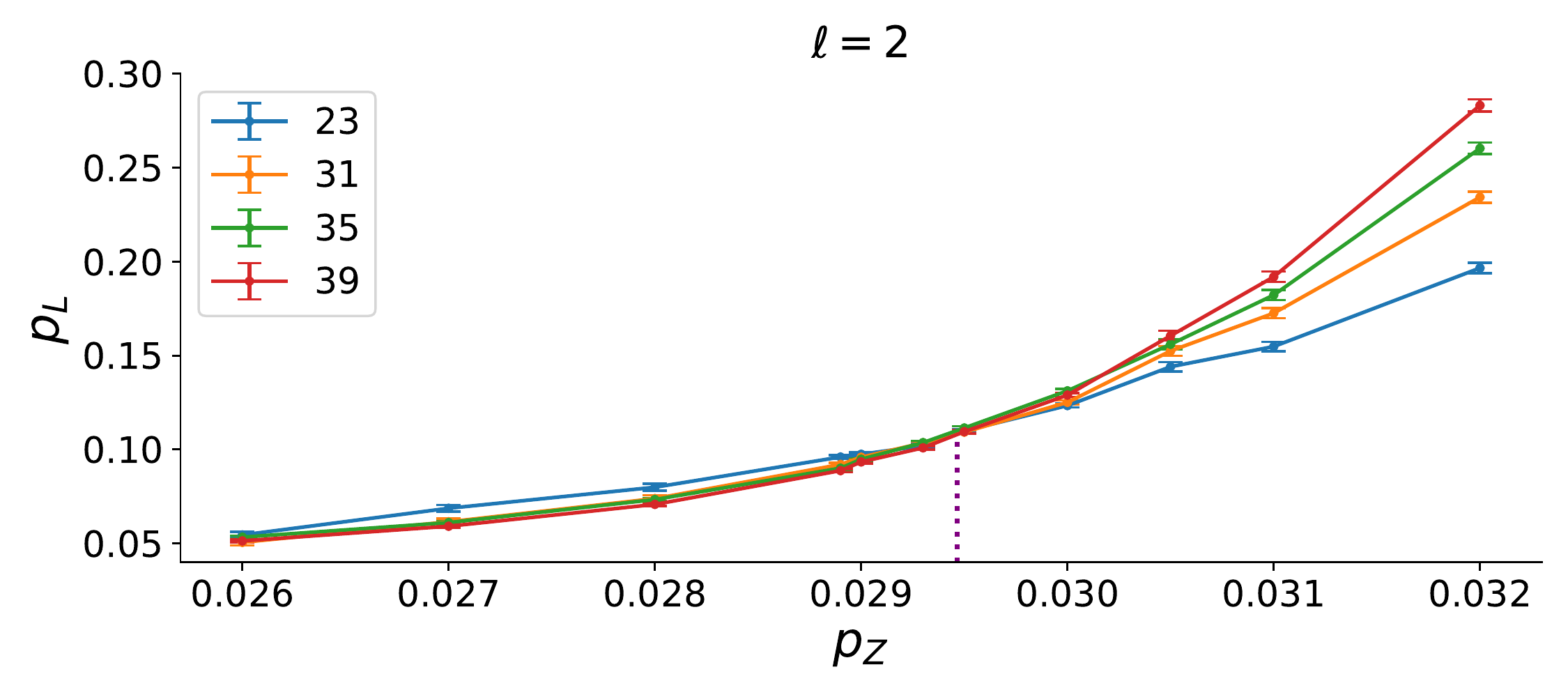}
\caption{Code capacity performance of the MWPM decoder for pure $Z$ noise for different \textit{levels} of FSC. Logical failure rate $p_L$ is plotted as a function of the physical Z error rate $p_Z$. The data points close to the threshold of the MWPM decoder were obtained with 100,000 Monte Carlo runs. Other points used 20,000 Monte Carlo runs.}
    \label{fig:MWPM_lf_plots}
\end{figure}

\begin{figure}[H]
    \centering
\sidesubfloat[]{\includegraphics[scale=0.2925]{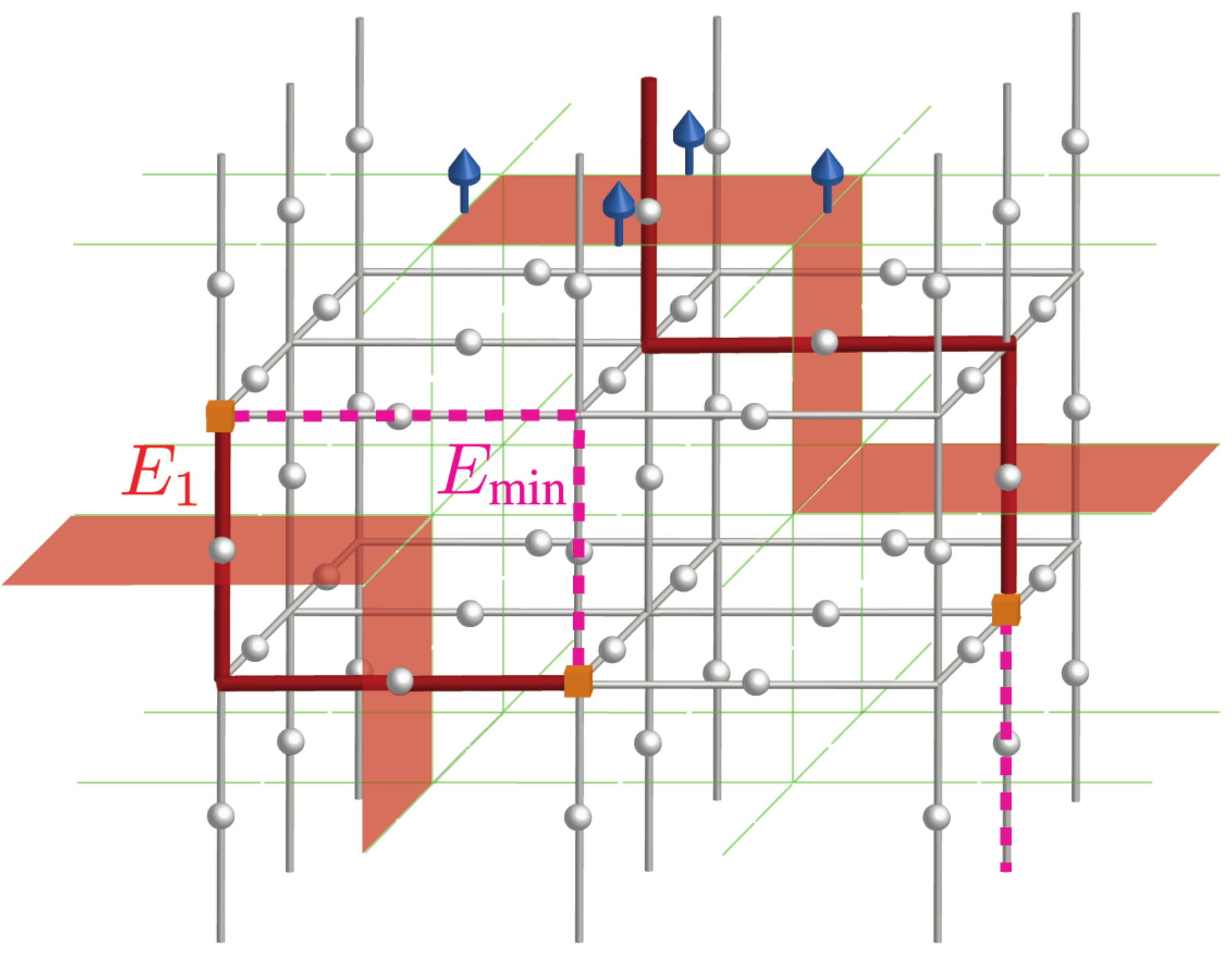}}
\sidesubfloat[]{\includegraphics[scale=0.525]{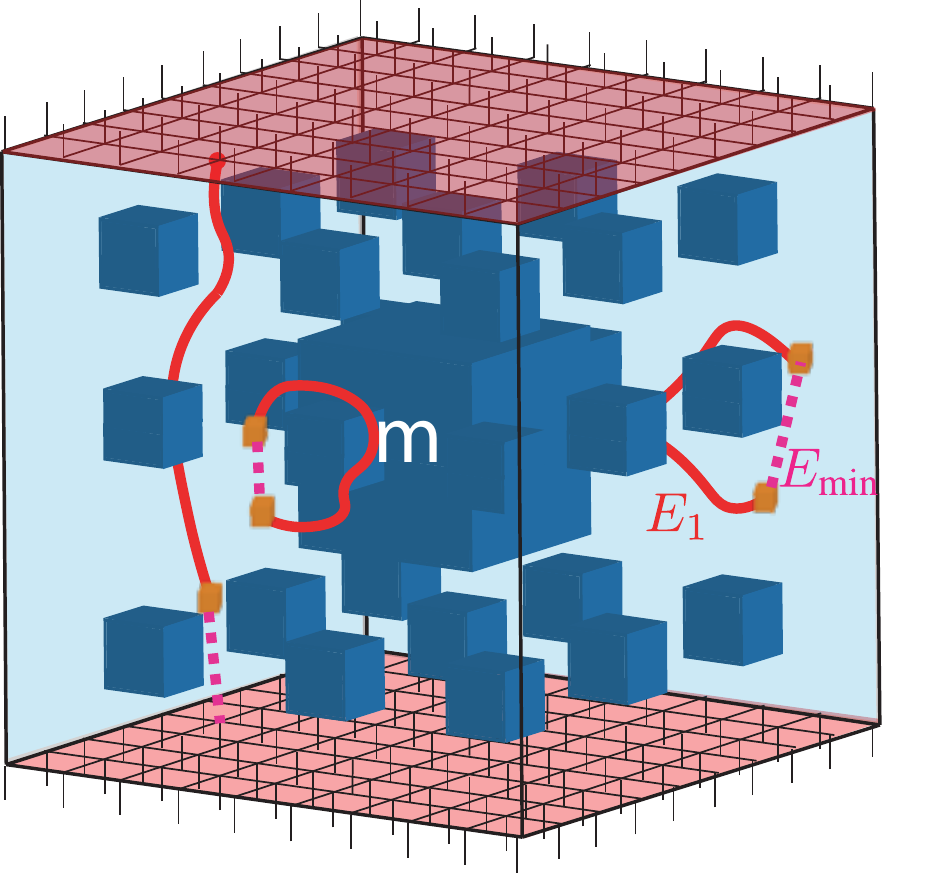}}
\caption{(a) Illustration of the RPGM. The gauge spins (blue) live on the edges of the dual lattice $\mathcal{L}^*$. Bold red edges represent the error chain $E_1$ (flux tubes) on the original lattice $\mathcal{L}$ which penetrates the ``wrong-sign" plaquettes (red) on $\mathcal{L}^*$. The endpoints of $E_1$ are monopoles (orange cube). The MWPM algorithm finds the recovery chain $E_\text{min}$ (pink dashed). A logical error occurs since $E_1+E_\text{min}$ contains a nontrivial relative cycle. (b) Illustration of the MWPM decoder on the fractal lattice $\widetilde{\mathcal{L}}$, where all flux tubes circumvent the $m$-holes.}
\label{fig:RPGM}
\end{figure}

\section{Phase transition on a fractal}
Quite interestingly, the code-capacity threshold of the MWPM decoder for the FSC can be mapped to the zero-temperature phase transition in a random-plaquette gauge model (RPGM) on a fractal lattice, using the statistical-mechanical mapping~\cite{Dennis_2002,wang2003confinement,BombinSM2009} of error correction in stabilizer codes. The disordered stat-mech Hamiltonian of the fractal RPGM can be written as:
\be
H=-\sum_{\mathsf{f}^*} \tau_{\mathsf{f}^*} U_{\mathsf{f}^*}, \quad \text{with} \quad U_{\mathsf{f}^*} = \prod_{\mathsf{e}^* \in \mathsf{f}^*} \sigma_{\mathsf{e}^*}.
\label{eqn:SM_RPGM}
\ee
Here, $U_{\mathsf{f}^*}$ is the $\mathbb{Z}_2$-valued ``gauge flux" penetrating the plaquette (face) $f^*$ on the dual fractal lattice $\widetilde{\mathcal{L}}^*$, whose value is determined by all the gauge Ising spins $\sigma_{\mathsf{e}^*} = \pm 1$ living on the links (edges) belonging to the plaquette $f^*$, as illustrated in Fig.~\ref{fig:RPGM}(a). Each Ising spin $\sigma_{\mathsf{e}^*}$ in the fractal RPGM is associated with a stabilizer on the edge $\widetilde{\mathsf{e}}^*$ in the code. Hence, the plaquette terms in Eq.~\eqref{eqn:SM_RPGM} are 4-body when in the bulk, while they can be 3-body or 2-body when they share an edge(s) with an $m$-boundary or two $m$-boundaries respectively. The random coupling strength $\tau_{f^*}=\pm 1$ on the plaquette $f^*$ represents the quenched disorder in this model. Its sign is determined by the $Z$ error on the qubit of the fractal code: with probability $p_Z$ an error occurs corresponding to the ``wrong-sign" choice $\tau_{f^*}=-1$ (favoring nontrivial flux $-1$) as shown in the highlighted plaquette (red) in Fig.~\ref{fig:RPGM}(a); with probability $1-p_Z$ no error occurs corresponding to the ``right-sign" choice  $\tau_{f^*}=+1$ (favoring trivial flux $+1$).  The collection of ``wrong-sign" plaquettes constitute flux tubes. 
Note that the flux tube can terminate on the $e$-boundary since the dual lattice $\widetilde{\mathcal{L}}^*$ contains plaquettes parallel to the $e$-boundary to allow the flux to penetrate through. At the same time, it is impossible for the flux tube to terminate on the $m$-boundary and hence $m$-holes, since no plaquette parallel to the $m$-boundary exists on $\widetilde{\mathcal{L}}^*$ [see in Fig.~\ref{fig:RPGM}(a)]. 

The flux tubes on $\widetilde{\mathcal{L}}^*$ correspond to the error chain $E_1$  (1-chain) living on the edges of the original lattice $\widetilde{\mathcal{L}}$, whose endpoints $S_0=\partial E_1$ are $\mathbb{Z}_2$ monopoles living on the cube $\mathsf{c}^*$ in $\widetilde{\mathcal{L}}^*$ (vertex $\mathsf{v}$ in $\widetilde{\mathcal{L}}$), corresponding to the syndromes ($e$-excitations) in the FSC. The MWPM decoder \cite{edmonds_1965} finds a recovery 1-chain $E_{\text{min}}$ on $\widetilde{\mathcal{L}}$ ending at the same set of monopoles (syndromes), \textit{i.e.}, $\partial E_{\text{min}}=S_0$, such that $E_{\text{min}}$ has the minimal length.  This is equivalent to finding $E_{\text{min}}$ such that the sum of disjoint closed flux tubes  $E_1 + E_{\text{min}}$ has minimal energy at zero temperature. On the other hand, the optimal (maximal-likelihood) decoder effectively finds the recovery chain $E_{\text{min}}$ with minimal free energy at finite temperature along the ``Nishimori line" \cite{wang2003confinement} and corresponds to an optimal threshold $p^*_\text{th}$.  In either case, the recovery succeeds if $E_{\text{min}}$ and $E_1$ belongs to the same homology class, \textit{i.e.}, $E_1 + E_{\text{min}}=\partial F$, where $F$ is a collection of faces on $\widetilde{\mathcal{L}}$, and fails if the closed flux tubes $E_1 + E_{\text{min}}$ contains a  homologically nontrivial relative cycle connecting the upper and lower $e$-boundaries.   In the Higgs (ordered) phase where the quenched disorder strength is low ($p_Z$$<$$p^*_\text{th}$), the failure probability is zero in the thermodynamic limit ($L$$\rightarrow$$\infty$) since the free energy cost of such a non-contractible relative cycle diverges linearly with cycle length and become unfavorable.  In contrast, in the confined (disordered) phase  ($p_Z$$>$$p^*_\text{th}$), the failure probability approaches one in the thermodynamic limit.   The error threshold $p_\text{th}$ of the MWPM decoder corresponds to the zero-temperature phase transition point between the Higgs phase and a gauge glass phase, which provides a lower bound of the confinement-Higgs transition occurring exactly at the optimal threshold $p^*_\text{th}$ along the ``Nishmori line" \cite{wang2003confinement}.
Note that the main difference from the case of the RPGM on a cubic lattice is that the flux tubes need to circumvent the $m$-holes [shown in Fig.~\ref{fig:MWPM_lf_plots}(b)].  Interestingly, the phase transition is tunable via the fractal (Hausdorff) dimension $D_H$.

\section{Discussion}
In this work, we proved that there exist decoders with nonzero thresholds for fractal surface codes (FSCs) on lattices with Hausdorff dimensions $D_H=2+\epsilon$. We noted that even in fractal dimensions, there exists a local decoder for the string-like syndromes. We also proved that the fault-tolerant MWPM threshold for the FSC is strictly greater than the same in the 3D surface code. Moreover, for a particular FSC with Hausdorff dimension $D_H\approx 2.966$, we demonstrated sweep decoder thresholds for bit-flip noise that are roughly the same as those for the regular 3D surface code. For the same FSC, we demonstrated code capacity MWPM threshold in the presence of phase-flip noise that is enhanced in comparison to the regular 3D surface code due to the suppression of the number of non-contractible cycles.  

The MWPM threshold provides a lower bound on the confinement-Higgs transition of the RPGM on a fractal lattice. In future work, one can continue studying the optimal threshold of a maximum-likelihood decoder and the exact value of the confinement-Higgs transition along the ``Nishimori line" \cite{wang2003confinement, Dennis_2002}.

In this work, we only studied the performance of the FSC as a quantum memory. Our work demonstrating thresholds on par with the 3D surface code motivates studies of the FSC for quantum computation. In upcoming work, we address implementation of the non-Clifford CCZ gate in a single-shot manner in the FSC, which implies that the space-time overhead scales like $\mathcal{O}(d^{2+\epsilon})$~\footnote{Manuscript in preparation}, which would be a fundamental space-time improvement on techniques such as those from Refs.~\cite{Bombin18,Brown_2020} which use 2D~lattices to simulate the action of 3D~topological codes for computation. 

Lastly, our work inspires further studies of quantum codes on fractal lattices embedded in three dimensions. For instance, the subsystem surface code~\cite{kubica2021singleshot} and the gauge color code~\cite{bombin2015gauge}, which gauge-fix to the 3D surface code and 3D color code, respectively, have properties of confinement and single-shot error correction, could be studied on fractal lattices. We leave this to a forthcoming work.

\textbf{Acknowledgements.---} 
We thank John Smolin and Andrew Cross for the valuable discussions and the help with the large-scale cluster simulations. AD thanks Michael Vasmer for valuable discussions and for sharing the sweep decoder data for the 3D surface code. AD thanks the IBM internship program for support, during which a part of this work was completed. AD is also supported by the Simons Foundation through the collaboration on Ultra-Quantum Matter (651438, AD) and by the Institute for Quantum Information and Matter, an NSF Physics Frontiers Center (PHY-1733907). TJO and GZ are supported by the U.S. Department of Energy, Office of Science, National Quantum Information Science Research Centers, and Co-design Center for Quantum Advantage (C2QA) under contract number DE-SC0012704. The code used for the simulations is publicly available at \href{https://github.com/dua-arpit/single_shot_decoding_fractal}{GitHub}.

\textbf{Author contributions.---}
All authors instigated the study and developed the decoding protocols. A.D. implemented the modified Sweep Decoder, while A.D. and G.Z. implemented the MWPM decoder. T.J., G.Z., and A.D. developed the proofs for the decoders. All authors contributed to the scientific discussion, running simulations, and writing of the manuscript.

\textbf{Corresponding authors.---}
Correspondence to any author. 
\bibliographystyle{quantum}
\bibliography{bib,bib2,bibtex_jochym}

\newpage
\newpage
\pagebreak

\appendix
\section{Proof of the fault-tolerant threshold of a matching decoder for correcting point-like syndromes}
\label{app:matching_proof}
In this section, we prove the existence of the fault-tolerant and code-capacity thresholds of the $Z$ errors in fractal surface codes, for the minimum-weight perfect matching (MWPM) decoder, following the methods from Ref.~\onlinecite{Dennis_2002}. In particular, we consider the presence of measurement errors under a phenomenological noise model: the Pauli-$Z$ error rate and measurement error rate are denoted by $p$ and $q$ respectively~\footnote{We suppress the subscript $Z$ of the notation $p_Z$ used for the $Z$-error rate in the main text}.  We also put bounds on the fault-tolerant and code-capacity thresholds of FSCs and show that they are strictly higher than the counterparts of the 3D surface code. Moreover, we show that in the asymptotic limit of $D_H=2+\epsilon$, both types of thresholds approach those of the 2D surface code. 

\subsection{Fault-tolerant MWPM thresholds}

The FSC is defined on a 3D fractal lattice $\widetilde{\mathcal{L}}^{(D_H)}$ embedded in (being a subset of) a 3D cubic lattice $\mathcal{L}^{(3)}$, \textit{i.e.}, $\widetilde{\mathcal{L}}^{(D_H)} \subset \mathcal{L}^{(3)}$, with $m$-holes being removed from $\mathcal{L}^{(3)}$. Here, $D_H$ represents the Hausdorff dimension of the fractal (space) lattice, and the tilde symbols will always indicate the fractal case from here on. In order to correct the measurement noise, we need to perform $d$ rounds~\cite{Dennis_2002} of error correction, where $d$ is the code distance. The corresponding space-time code is defined on the space-time lattice $\widetilde{\mathcal{L}}^{(D_H)} \times \mathsf{l}_t$ embedded in a 4D hypercubic lattice $\mathcal{L}^{(4)}\equiv \mathcal{L}^{(3)} \times \mathsf{l}_t$. Here, $\mathsf{l}_t$ represents a 1D lattice along the time direction with $d$ edges and $d+1$ vertices. 

On the 4D space-time lattice $\widetilde{\mathcal{L}}^{(D_H)}\times \mathsf{l}_t$, a qubit error event occurs on a space edge $l_\text{s.p.}$ at a certain time step $t$. These space edges are just edges belonging to the space lattice $\widetilde{\mathcal{L}}^{(D_H)}$. On the other hand, a syndrome ($X$-stabilizer measurement)  is located at a vertex $(v,t)$ on the space-time lattice, where $v$ corresponds to the vertex label on the space lattice $\mathcal{L}_{FSC}$ and $t$ labels the time step.  This syndrome is hence denoted by $s(v, t) \in \{0,1\}$.  Now we define the modified syndrome at time $t$ to be the difference between syndromes at time $t+1$ and time $t$ on the same vertex $v$ on the space lattice, \textit{i.e.}, $s'(v, t) = s(v, t+1) - s(v, t)$.  For the special case of $t=d$, we assign $s'(v, d) = s(v, d)$.   Therefore, a measurement error at time $t$ and vertex $v$ of the space lattice corresponds to a time edge $l_\text{T}$ connecting the space-time vertex $(v, t+1)$ and $(v, t)$.  For example, if a single measurement error occur on this time edge $l_\text{T}$, both its neighboring  syndromes $s'(v, t)$ and $s'(v, t+1)$ will be highlighted, \textit{i.e.}, $s'(v, t)=s'(v, t+1)=1$.   

As we see, both the qubit and measurement errors occur on the edges (1-cells) of the 4D space-time lattice. In the absence of measurement errors ($q=0$), qubit errors just reside on the edges of the 3D space lattice. Therefore, we can describe generic errors by an error chain $E$, which mathematically corresponds to a 1-chain of the cell complex, \textit{i.e.}, the 4D or 3D lattice \footnote{From now on, we suppress the subscript of the error chain $E_1$ used in the main text for clarity.}. The error chain $E$ is characterized by a function $n_E(e)$ that takes an edge $e$ to the $\mathbb{Z}_2$-coefficient $n_E(e) \in \{0, 1\}$,  where $n_E(e)=1$ corresponds to the edge being occupied by the error chain $E$. For simplicity, we just consider the isotropic case that qubit and measurement error rates are the same, \textit{i.e.}, $p=q$.  The probability  that error chain $E$ occurs is 
\begin{align}\label{eq:chain_probability}
\nonumber \text{Pr}(E) =& \prod_e (1-p)^{1-n_E}  \\
 =& \bigg[\prod_e (1-p)\bigg] \prod_e\left(\frac{p}{1-p}\right)^{n_E(e)}.
\end{align}
The boundary of the error chain $E$ is the collection of highlighted modified syndromes on the vertices (0-cells) denoted by $S'$, \textit{i.e.}, $S'=\partial E$, which is a 0-chain of the cell complex.  

Given the measured  syndrome information $S'$, the MWPM decoder needs to guess a recovery chain $E_\text{min}$ which is in the same homology class as the actual error chain $E$. First, the recovery chain $E_\text{min}$ should also have the highlighted syndromes $S'$ as its boundary, \textit{i.e.},
\begin{equation}
\partial E_\text{min} = S' = \partial E.   
\end{equation}
Moreover, the recovery succeeds if $E_\text{min}$ is in the same homology class as $E$, which means the following condition needs to be satisfied:
\begin{align}
E + E_\text{min} = \partial F,
\end{align}
where the 1-chain $E + E_\text{min}$ is a cycle, \textit{i.e.}, with no boundary: $\partial (E + E_\text{min} )=0$. In addition, this 1-cycle $E + E_\text{min} $ also needs to be the boundary of a collection of faces $F$ (2-cells), \textit{i.e.}, $\partial F$, and is hence contractible, \textit{i.e.}, homologically trivial. The recovery fails if $E + E_\text{min}$ contains homologically non-trivial cycles which wrap around a 3-torus or connect two different $e$-boundaries on $\widetilde{\mathcal{L}}^{(D_H)}$. From now on, we focus on the case of a 3-torus, \textit{i.e.}, with periodic boundary condition.  The threshold with external $e$-boundaries is expected to be similar.  

Apart from an overall normalization, the error chain $E$ occurs with probability $(\frac{p}{1-p})^{|E|}$ according to Eq.~\eqref{eq:chain_probability}, where $|E|$ denotes the total number of edges on the error chain $E$, \textit{i.e.}, the chain length. The MWPM decoder aims to find $E_\text{min}$ that maximizes this probability, which is equivalent to minimizing $|E| \log(\frac{1-p}{p})$ and effectively the chain length $|E|$ for fixed $p$.

Now we consider bounding the likelihood of homologically non-trivial cycles being contained in $E$$+$$E_\text{min}$, which corresponds to the logical failure rate of the MWPM decoder.  

We consider a particular cycle $C$ on the space-time lattice   $\widetilde{\mathcal{L}}^{(D_H)}\times \mathsf{l}_t$ with $|C| \equiv l$ edges.  The actual error chain $E$ contains $|E|$ edges, and the estimated recovery chain $E_\text{min}$ from the MWPM algorithm contains $|E_\text{min}|$ edges. We then ask the probability that the cycle $C$ is contained in $E+E_\text{min}$.

First, we can get the following inequality for the chain length:   
\begin{equation}
|E| + |E_\text{min}| \ge |E+E_\text{min}| \ge |C| \equiv l.	
\end{equation}
The first $``\ge"$ is due to the fact that some of the edges of $E$ and $E_\text{min}$ can overlap and cancel at their binary ($\mathbb{Z}_2$) sum $E + E_\text{min}$, while the equality holds when there is no overlap between $E$ and $E_\text{min}$, \textit{i.e.}, $\textbf{supp}(E) \cap \textbf{supp}(E_\text{min})=0$, where $\textbf{supp}(E)$ denotes the support of the error chain $E$, \textit{i.e.}, the set of edges $\{l\}$ with $n_E(l)=1$. The second $``\ge"$ simply attributes to the fact that $C$ is a subset of $E+E_\text{min}$. 
Therefore, we can get the following inequality between the chain probability (up to a normalization constant) of $E$, $E_\text{min}$  and $C$:
\begin{align}\label{eq:chain_inequality1}
    \left(\frac{p}{1-p}\right)^{|E_\text{min}|} \left(\frac{p}{1-p}\right)^{|E|} \le \left(\frac{p}{1-p}\right)^{l}.
\end{align}
Moreover, since we have taken $E_\text{min}$ to be the minimal-length chain with the boundary being the highlighted syndrome set $S'$, we have $|E_\text{min}| \le |E|$ and hence 
\begin{align}\label{eq:chain_inequality2}
    \left(\frac{p}{1-p}\right)^{|E|} \le \left(\frac{p}{1-p}\right)^{|E_\text{min}|}.
\end{align}
The combination of Eq.~\eqref{eq:chain_inequality1} and Eq.~\eqref{eq:chain_inequality2} leads to the following inequality:
\begin{align}\label{eq:chain_inequality3}
    \left(\frac{p}{1-p}\right)^{|E|} \le \left[\left(\frac{p}{1-p}\right)^{l}\right]^{\frac{1}{2}}.
\end{align}

Now we consider the probability $P(l)$ that a particular cycle with $l$ edges is contained in $E+E_\text{min}$. There are in total $2^{l}$ ways to distribute errors (edges contained in $E$) at locations on the specified chain since each edge either has the error or not.  For specified error locations, the probability for these errors to occur is 
\begin{align}
p^{|E|}(1-p)^{l-|E_C|} = (1-p)^{l}\left(\frac{p}{1-p}\right)^{|E_C|}. 
\end{align}
where $E_C$ is the part of $E$ that overlaps with $C$. Therefore, we get the following probability 
\begin{align}\label{eq:chain_inequality4}
\text{Pr}(l) &= 2^{l} (1-p)^{l}\left(\frac{p}{1-p}\right)^{|E_C|}\nonumber \\
&\le 2^{l} (1-p)^{l}\left(\frac{p}{1-p}\right)^{|E|}\le 2^{l} p^{\frac{l}{2}}(1-p)^{\frac{l}{2}}, 
\end{align}
where the inequality follows from  Eq.~\eqref{eq:chain_inequality3}.  
 
Next, we can bound the probability of any cycle with $l$ edges contained in the chain $E+E_\text{min}$ via path counting.  The cycle $C$ can be considered as a walk on a lattice which begins and ends at a randomly chosen point on the cycle.  We need to estimate the likelihood that the closed walk is homologically nontrivial (non-contractible). The walks corresponds to connected chain of errors and visit any given edge at most once. It will be convenient to further restrict the walks to be self-avoiding walks (SAWs) which visit any given vertex at most once, with the exception of the starting/ending point which is revisited.  Given any homologically non-trivial closed walk, one can obtain a closed SAW (self-avoiding polygon: SAP) by eliminating some homologically trivial cycles from the walk, which does not change the homology of the cycle and hence the presence or absence of the logical error. 

In order to estimate the logical error rate, we consider SAPs lying between two time slices separated by time steps $T$, and assume $T=O(d)$, where $d=L$ is the code distance which also equals the linear system size in this code family. We denote the number of SAPs with $H$ space (horizontal) edges and $V$ time (vertical) edges by $N_\text{SAP}(H,V)$, and one can express the total number of edges as the sum of the number of the two types of edges: $|C|=H+V$. A self-avoiding random walk can start at any of the $d^{n-1} T$ sites on the $n$-dimensional space-time lattice. In the case of SAP, the starting point can be chosen to be any of the points on the SAP. The probability that $E+E_\text{min}$ contains any SAP with $H$ space edges and $V$ time edges obeys the following inequality:
\begin{equation}
\text{Pr}_\text{SAP}(H, V) \le d^{n-1} T N_\text{SAP}(H, V)2^{H+V} p^{\frac{H+V}{2}}(1-p)^{\frac{H+V}{2}}.    
\end{equation}
 The minimal homologically nontrivial cycle needs to contain at least $d$ space (horizontal) edges, \textit{i.e.}, $H \ge d$.  Therefore, we can bound the logical failure  rate as
\begin{align}\label{eq:failure_probability}
\nonumber \text{Pr}_\text{fail} &\le \sum_V \sum_{H \ge d} \text{Pr}_\text{SAP}(H,V) \\ 
&\le d^{n-1} T \sum_V \sum_{H \ge d} N_\text{SAP}(H,V)[4p(1-p)]^{\frac{H+V}{2}}.
\end{align}

Now we aim to obtain a bound for the fault-tolerant threshold. The number of SAPs does not distinguish the difference between space (horizontal) and time (vertical) edges, and we hence have $N_\text{SAP}(H,V) \equiv N^{(n)}_\text{SAP}(l)$, where $n$ denotes the dimension of the space-time lattice and $l$ denotes the total number of edges on the SAP. For a (hyper)cubic lattice in $n$ dimensions, the first step of the SAP can choose any of the $2n$ directions,  while the subsequent steps will have at most $2n-1$ directions due to self avoidance.  Therefore, one can get the following naive bound
\begin{equation}
N^{(n)}_\text{SAP}(l) \le 2n(2n-1)^{l-1}.    
\end{equation}
However, there exist tighter bounds, found using results in self-avoiding polygons for the (hyper)cubic lattice in $n=2$, $n=3$ and $n=4$ \cite{Dennis_2002}:
\begin{align}
N^{(2)}_\text{SAP}(l) &\le P_2(l)(\mu_2)^l, \quad \mu_3 \approx 2.638, \\ %\label{eq:SAP3D1} 
N^{(3)}_\text{SAP}(l) &\le P_3(l)(\mu_3)^l, \quad \mu_3 \approx 4.684,  \\ %\label{eq:SAP3D2}  
N^{(4)}_\text{SAP}(l) &\le P_4(l)(\mu_4)^l, \quad \mu_4 \approx 6.77,     \label{eq:SAP4D}
\end{align}
where $P_2(l)$, $P_3(l)$ and $ P_4(l)$ are polynomials.  

In the fault-tolerant case, we first consider the 4D hypercubic space-time lattice $\mathcal{L}^{(4)}\equiv \mathcal{L}^{(3)} \times \mathsf{l}_t$, and we can hence substitute Eq.~\eqref{eq:SAP4D} into 
Eq.~\eqref{eq:failure_probability} and obtain
\begin{equation}\label{eq:failure_probability2}
\text{Pr}^{(3+1)}_\text{fail} \le  d^3 T \sum_V \sum_{H \ge d} P_4(H+V) [4\mu_4^2p(1-p)]^{\frac{H+V}{2}}.
\end{equation}
Using $H+V \ge d$ and imposing the following condition 
\begin{equation}\label{eq:threshold_condition}
p(1-p) < (4\mu_4^2)^{-1},
\end{equation} we obtain the following inequality
\begin{align}\label{eq:failure_probability3}
\nonumber \text{Pr}^{(3+1)}_\text{fail} &\le  d^3 T \sum_V \sum_{H \ge d} P_4(H+V) [4\mu_4^2p(1-p)]^{d/2} \\
& < Q_4(d, T) [4 \mu_4^2 p(1-p) ]^{d/2},
\end{align}
where $Q_4(d, T)$ is some polynomial of $d$ and $T$.  The second inequality in Eq.~\eqref{eq:failure_probability3} comes from the fact that there are in total $2d^3 T$ space edges and $d^3 T$ time edges in the 4D hypercubic space-time lattice, so there are at most $2d^6 T$ terms in the sum in the first line of Eq.~\eqref{eq:failure_probability3} which is absorbed into $Q_4(d, T)$.  Therefore, as long as Eq.~\eqref{eq:threshold_condition} is satisfied, the logical failure rate $\text{Pr}_\text{fail}$ decays exponentially with the code distance $d$ and hence approaches 0 in the thermodynamic limit $d \rightarrow \infty$. One can then obtain an analytic lower bound of the fault-tolerant threshold of the usual 3D surface code based on Eq.~\eqref{eq:threshold_condition} \cite{Dennis_2002}:
\begin{equation}\label{eq:bound_4D}
  p^{(3+1)}_\text{th} > 0.00548.  
\end{equation}

Now we switch to the case of the fractal surface code defined on the space-time lattice $\widetilde{\mathcal{L}}^{(D_H)} \times \mathsf{l}_t$ which is embedded in the 4D hypercubic space-time lattice $\mathcal{L}^{(4)}$$\equiv$$\mathcal{L}^{(3)}\times \mathsf{l}_t$.   The only difference from the hypercubic case  is that we need to replace $N^{(4)}_\text{SAP}(l)$ in Eq.~\eqref{eq:failure_probability} with  $\widetilde{N}^{(D_H+1)}_\text{SAP}(l)$, which represents the total number of self-avoiding polygon on the fractal space-time lattice with Hausdorff dimension $D_H+1$.   Since the fractal space-time lattice is just a subset of the 4D hypercubic lattice, \textit{i.e.}, $\widetilde{\mathcal{L}}^{(D_H)} \times \mathsf{l}_t 
\subset \mathcal{L}^{(4)}$ with edges in the hole regions being removed, some SAPS present in the hypercubic lattice is hence removed in the fractal lattice case. Therefore, the number of SAPs in the fractal space-time lattice must be strictly smaller than the SAPs in the 4D hypercubic lattice, \textit{i.e.}, 
\begin{equation}\label{eq:SAP_inequality}
\widetilde{N}^{(D_H+1)}_\text{SAP}(l) < N^{(4)}_\text{SAP}(l).
\end{equation}
Based on Eq.~\eqref{eq:SAP_inequality} and the bound in Eq.~\eqref{eq:failure_probability}, one gets the following bound:
\begin{equation}\label{eq:failure_bound}
\widetilde{\text{Pr}}^{(D_H + 1)}_\text{fail} < \text{Pr}^{(3+1)}_\text{fail}. 
\end{equation}
where $\widetilde{\text{Pr}}^{(D_H + 1)}_\text{fail}$ and $\text{Pr}^{(3+1)}_\text{fail}$ represent the logical failure rate of the fractal surface code with Hausdorff dimension $D_H$ and the 3D surface code respectively in the fault-tolerance context.   Since one has  $\text{Pr}^{(3 + 1)}_\text{fail} \rightarrow 0$ in the thermodynamic limit ($d\rightarrow \infty$) and  below the fault-tolerant error threshold, \textit{i.e.},  $p<p^{(3+1)}_\text{th}$,  one will also get $\widetilde{\text{Pr}}^{(D_H + 1)}_\text{fail} \rightarrow 0$ in the thermodynamic limit due to the bound in Eq.~\eqref{eq:failure_bound}.   This proves that the fractal surface code is also fault-tolerant below a certain threshold. Furthermore, the fault-tolerant threshold for the fractal surface code must be strictly larger than the 3D surface code, \textit{i.e.},
\begin{equation}\label{eq:threshold_bound}
\widetilde{p}^{(D_H + 1)}_\text{th} > p^{(3+1)}, 
\end{equation}
since for any $p<p^{(3+1)}$ one always has $\widetilde{\text{Pr}}^{(D_H + 1)}_\text{fail} \rightarrow 0$ in the thermodynamic limit. When the Hausdorff dimension of the fractal code is asymptotically approaching 3D, the threshold will also approach to that of the 3D surface code, \textit{i.e.}, $\lim_{D_H \rightarrow 3} \widetilde{p}^{(D_H + 1)}_\text{th} \rightarrow p^{(3+1)} $.  
According to the derived bound in Eq.~\eqref{eq:bound_4D}, we can also analytically bound the error threshold for the fractal surface code, \textit{i.e.}, 
 \boxedeq{eq:LBFTDHa}{\widetilde{p}^{(D_H + 1)}_\text{th} > 0.00548 \quad    \text{(analytic)}.}

% \begin{equation}
% \widetilde{p}^{(D_H + 1)}_\text{th} > 0.00548 \quad    \text{(analytic)}.
% \end{equation}
Meanwhile, existing numerical simulation in the literature gives an estimation of the fault-tolerant threshold
of the 3D surface code  $p^{(3 + 1)} \approx 0.0125$ \cite{vasmer2019fault_thesis} under the phenomenological noise model. Therefore, according to the bound in Eq.~\eqref{eq:threshold_bound}, we can improve the lower bound for the fractal code under phenomenological noise to be
 \boxedeq{eq:LBFTDHn}{\widetilde{p}^{(D_H + 1)}_\text{th} > 0.0125  \quad \text{(numerical)}.}

% \begin{equation}
% \widetilde{p}^{(D_H + 1)}_\text{th} > 0.0125  \quad \text{(numerical)}.
% \end{equation}
Next, we consider the upper bound and asymptotic limit of the family of fractal surface codes with Hausdorff dimension $D_H$ asymptotically approaching $2+\epsilon$. We first consider the fault-tolerant threshold of a 2D surface code, which  effectively corresponds to a matching problem on a 3D cubic space-time lattice $\mathcal{L}^{(3)}\equiv \mathcal{L}^{(2)} \times \mathsf{l}_t$, where $\mathcal{L}^{(2)}$ represents the 2D space lattice.  Similar to the derivation of  Eq.~\eqref{eq:failure_probability3},  we can obtain the following bound:
\begin{equation}
    \text{Pr}^{(2+1)}_\text{fail} < Q_3(d, T) [4 \mu_3^2 p(1-p) ]^{d/2},
\end{equation}
given the following condition:
\begin{equation}\label{eq:threshold_condition_2}
p(1-p) < (4\mu_3^2)^{-1}. 
\end{equation}
Here, $Q_3(d, T)$ is some polynomial of $d$ and $T$.
This in turn leads to the following analytic bound on the fault-tolerant threshold of the 2D surface code.
\begin{equation}\label{eq:bound_3D}
  p^{(2+1)}_\text{th} > 0.0114.  
\end{equation}
Now when we gradually reduce the fractal dimension $D_H$ in the family of fractal surface codes defined on the fractal cube geometry, the fractal space-time lattice $\widetilde{\mathcal{L}}^{D_H}\times \mathsf{l}_t$  is asymptotically approaching a 3D cubic lattice, \textit{i.e.}, 
\begin{equation}
  \widetilde{\mathcal{L}}^{(2+\epsilon)}\times \mathsf{l}_t \rightarrow  \mathcal{L}^{(2)}\times \mathsf{l}_t   = \mathcal{L}^{(3)},
\end{equation}
which means all the quantitative properties of the fractal lattice is asymptotically approaching those of the cubic lattice, including the scaling of the number of SAP:
\begin{equation}\label{eq:asymptotic}
    \lim_{\epsilon \rightarrow 0} \widetilde{N}^{(2+\epsilon+1)}_\text{SAP}(l) \sim N^{(3)}_\text{SAP}(\ell)
\end{equation}
This in turn leads to the following asymptotic threshold: 
\begin{equation}
\lim_{\epsilon \rightarrow 0}  \widetilde{p}^{(2+\epsilon+1)}_\text{th} = p^{(2+1)}_\text{th},
\end{equation}
and the upper bound for the whole family of fractal surface code
\begin{equation}\label{eq:upper_bound}
\widetilde{p}^{(D_H+1)}_\text{th} < p^{(2+1)}_\text{th}.
\end{equation}
According to Eq.~\eqref{eq:bound_3D} and \eqref{eq:asymptotic}, we can get the following strict bound on the asymptotic threshold:
% \begin{equation}
% \lim_{\epsilon \rightarrow 0}  \widetilde{p}^{(2+\epsilon+1)}_\text{th} > 0.0114 \quad  \text{(analytic)}.
% \end{equation}
 \boxedeq{eq:LBFTDH2a}{\lim_{\epsilon \rightarrow 0}  \widetilde{p}^{(2+\epsilon+1)}_\text{th} > 0.0114 \quad  \text{(analytic)}.}

Meanwhile, existing numerical MWPM simulation gives an estimation of the fault-tolerant threshold of the 2D surface code $p^{(2 + 1)} \approx 0.029$ \cite{Dennis_2002, wang2003confinement} under the phenomenological noise model. 
Therefore, we can get an estimation based on the numerical value as:
 \boxedeq{eq:LBFTDH2n}{\lim_{\epsilon \rightarrow 0}  \widetilde{p}^{(2+\epsilon+1)}_\text{th} \approx 0.029 \quad  \text{(numerical)}.}
% \begin{equation}
% \lim_{\epsilon \rightarrow 0}  \widetilde{p}^{(2+\epsilon+1)}_\text{th} \approx 0.029 \quad  \text{(numerical)}.
% \end{equation}
In sum, we have shown that the fault-tolerant threshold of the fractal surface codes under pure Pauli-$Z$ and measurement errors is strictly higher than the 3D surface code, and asymptotically approaching the value of the 2D surface code, which implies a significant practical advantage due to its additional ability of performing logical CCZ gate. 

\subsection{Code capacity MWPM thresholds}
Finally, we also discuss the code-capacity error threshold, \textit{i.e.}, in the absence of measurement error ($q=0$),  which is numerically studied in the main text.  In this case, we only consider the space lattice.  For the 3D and 2D surface codes, the thresholds correspond to the matching problems on a 3D lattice $\mathcal{L}^{(3)}$ and 2D lattice $\mathcal{L}^{(2)}$ respectively.  Similar to the derivation of  Eq.~\eqref{eq:failure_probability3},  we can obtain the following bounds:
\begin{align}
    \text{Pr}^{(3)}_\text{fail} <& Q_3(d, T) [4 \mu_3^2 p(1-p) ]^{d/2}, \\
    \text{Pr}^{(2)}_\text{fail} <& Q_2(d, T) [4 \mu_2^2 p(1-p) ]^{d/2},
\end{align}
given the following conditions respectively:
\begin{align}
p(1-p) <&  (4\mu_3^2)^{-1},  \\
p(1-p) <&  (4\mu_2^2)^{-1}. 
\end{align}
We hence get the following analytic bound on the code-capacity error thresholds for 3D and 2D surface codes:
\begin{align}
  p^{(3)}_\text{th} >& 0.0114,   \\
  p^{(2)}_\text{th} >& 0.0373. 
\end{align}
Similar to Eq.~\eqref{eq:threshold_bound} and Eq.~\eqref{eq:upper_bound} and the corresponding arguments in the fault-tolerant scenario, we obtain the upper and lower bound of the code-capacity error threshold of the family of fractal codes with Hausdorff dimension $2<D_H<3$ in terms of the thresholds of the 2D and 3D surface codes, \textit{i.e.}, 
\begin{equation}
    p^{(3)}_\text{th} < \widetilde{p}^{(D_H)}_\text{th} < p^{(2)}_\text{th}.
\end{equation}
We hence get the following analytic lower bound:
 \boxedeq{eq:LBCCDHa}{\widetilde{p}^{(D_H)}_\text{th} > 0.0114 \quad    \text{(analytic)}}
% \begin{equation}
% \widetilde{p}^{(D_H)}_\text{th} > 0.0114 \quad    \text{(analytic)},
% \end{equation}
and the analytic bound on the asymptotic code-capacity threshold at $D_H=2+\epsilon$:
 \boxedeq{eq:LBCCDH2a}{\lim_{\epsilon \rightarrow 0}  \widetilde{p}^{(2+\epsilon)}_\text{th} > 0.0373 \quad  \text{(analytic)}.}
% \begin{equation}
% \lim_{\epsilon \rightarrow 0}  \widetilde{p}^{(2+\epsilon)}_\text{th} > 0.0373 \quad  \text{(analytic)}.
% \end{equation}
Meanwhile, existing numerical simulation gives the estimate of the code-capacity threshold of the 3D surface code $p^{(3)}\approx 0.029$ under the phenomenological noise model \cite{wang2003confinement, vasmer2019fault_thesis}, which is exactly the same as the fault-tolerant threshold of the 2D surface code $p^{(2+1)}$ mentioned above.  On the other hand, the numerical estimate of the code-capacity threshold of the 2D surface code is $p^{(2)}\approx 0.1031$ \cite{wang2003confinement}.  Therefore, we get the following numerical lower bound of the class of fractal surface codes:
 \boxedeq{eq:LBCCDHn}{\widetilde{p}^{(D_H)}_\text{th} > 0.029 \quad    \text{(numerical)},}
% \begin{equation}
% \widetilde{p}^{(D_H)}_\text{th} > 0.029 \quad    \text{(numerical)},
% \end{equation}
and the numerical estimation on the asymptotic code-capacity threshold:
% \begin{equation}
% \lim_{\epsilon \rightarrow 0}  \widetilde{p}^{(2+\epsilon)}_\text{th} \approx 0.1031 \quad  \text{(numerical)}.
% \end{equation}
 \boxedeq{eq:LBCCDH2n}{\lim_{\epsilon \rightarrow 0}  \widetilde{p}^{(2+\epsilon)}_\text{th} \approx 0.1031 \quad  \text{(numerical)}.}
Indeed, our numerical simulation on the fractal $FC(3,1)$ in the main text shows that $\widetilde{p}^{(D_H=2.966)}_\text{th} > p^{(3)}_\text{th} \approx 0.029$.

\section{Counting argument for topological degeneracy from $m$-holes}
\label{app:gsd_mhole_counting}
In the main text, we presented a homological argument to show that the $m$-holes do not encode any logical information. Here, we confirm that via a simple counting argument. 

A single $L\times L\times L$ $m$-hole changes the number of $Z$ stabilizers $N_Z$ to $N_Z^\prime=N_Z-3L^2(L-1)$ since there are $3(L-1)$ planes from which $L^2$ plaquette operators are removed. The number of $X$ stabilizers changes from $N_X$ to $N_X^\prime=N_X-(L-1)^3$ because $(L-1)^3$ operators lose all their qubits inside the hole. The number of relations among the $Z$ and $X$ stabilizers also changes from $R_Z$ and $R_X$ to $R_Z^\prime=R_Z-(L^3-1)$ and $R_X^\prime=R_X$ respectively. The $L^3-1$ comes from the relations on $L^3$ cubes that are removed but leaving a single relation coming from plaquette operators on the surface of the $m$-hole. The number of physical qubits changes from $N_q$ to $N_q^\prime=N_q-3L(L-1)^2$ since there are $3L$ planes, each with $(L-1)^2$ physical qubits, inside the $L\times L\times L$ $m$-hole. Thus, we have the number of encoded qubits $k^\prime$ as 
\begin{align*}
k^\prime&=N_q^\prime -(N_S^\prime-R^\prime)\\
&=N_q^\prime -(N_Z^\prime+N_X^\prime-R_Z^\prime-R_X^\prime)\\
&=N_q-3L(L-1)^2\\
&~-(N_Z-3L^2(L-1)+N_X-(R_Z-L^3+1)-R_X)\\
&=N_q-(N_S-R)\\
&=k
\end{align*}
where $N_S=N_Z+N_X$, $R=R_Z+R_X$ and $k$ is the number of encoded qubits before making the $m$-hole. The argument generalizes to $L_x\times L_y \times L_z$ $m$-holes and an arbitrary number of them. 
\newpage
\section{Sweep decoder threshold plots for rounds $N=0,32$}
\label{app:sweep_thresholds_lowN}
\vspace{-3mm}

Below, we show the threshold plots for levels $\ell=1,2$ of the FSC for the number of rounds of measurements as $N=1,33$. 
\vspace{-3mm}

\begin{figure}[H]
    \centering
    \vspace{2mm}
    \sidesubfloat[]{ \includegraphics[width=0.63\columnwidth]{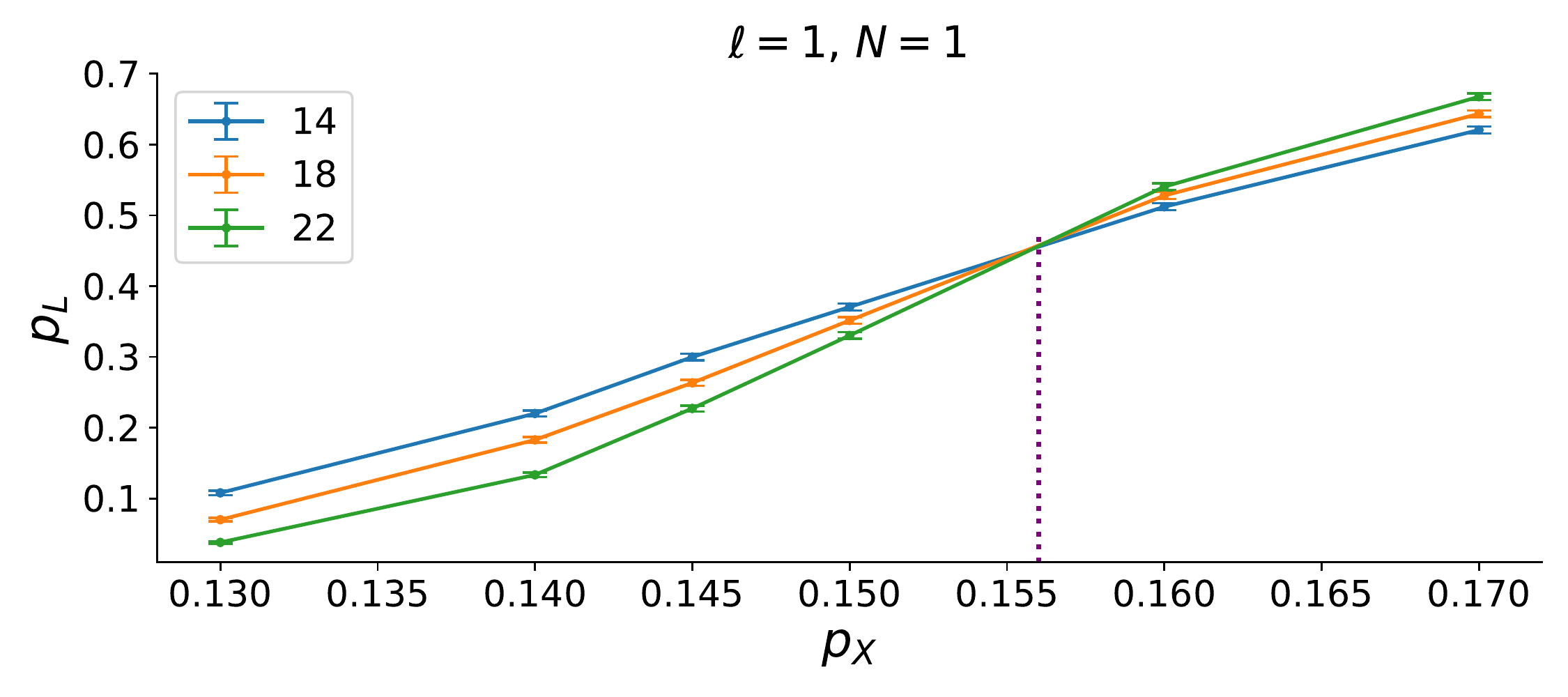}}\\
   \sidesubfloat[]{ \includegraphics[width=0.63\columnwidth]{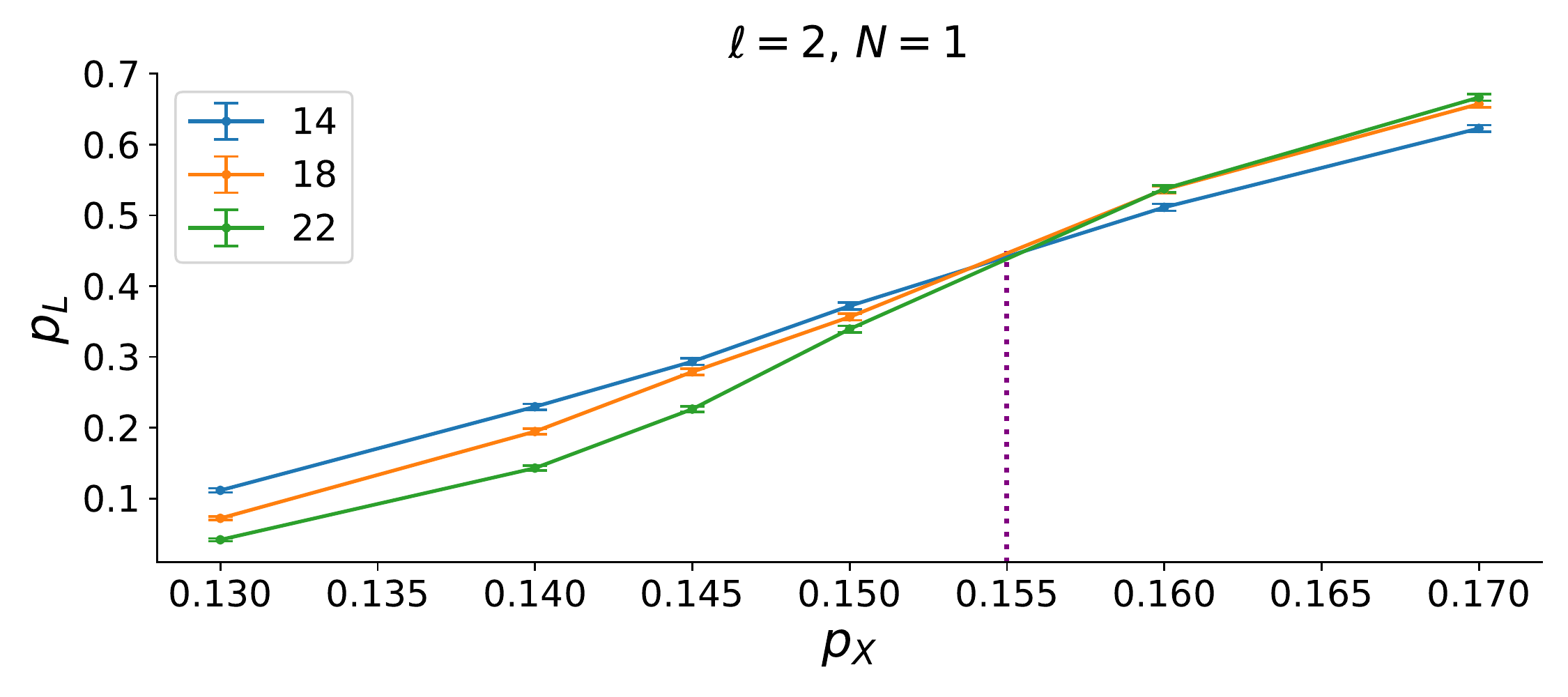}}\\
   \sidesubfloat[]{ \includegraphics[width=0.63\columnwidth]{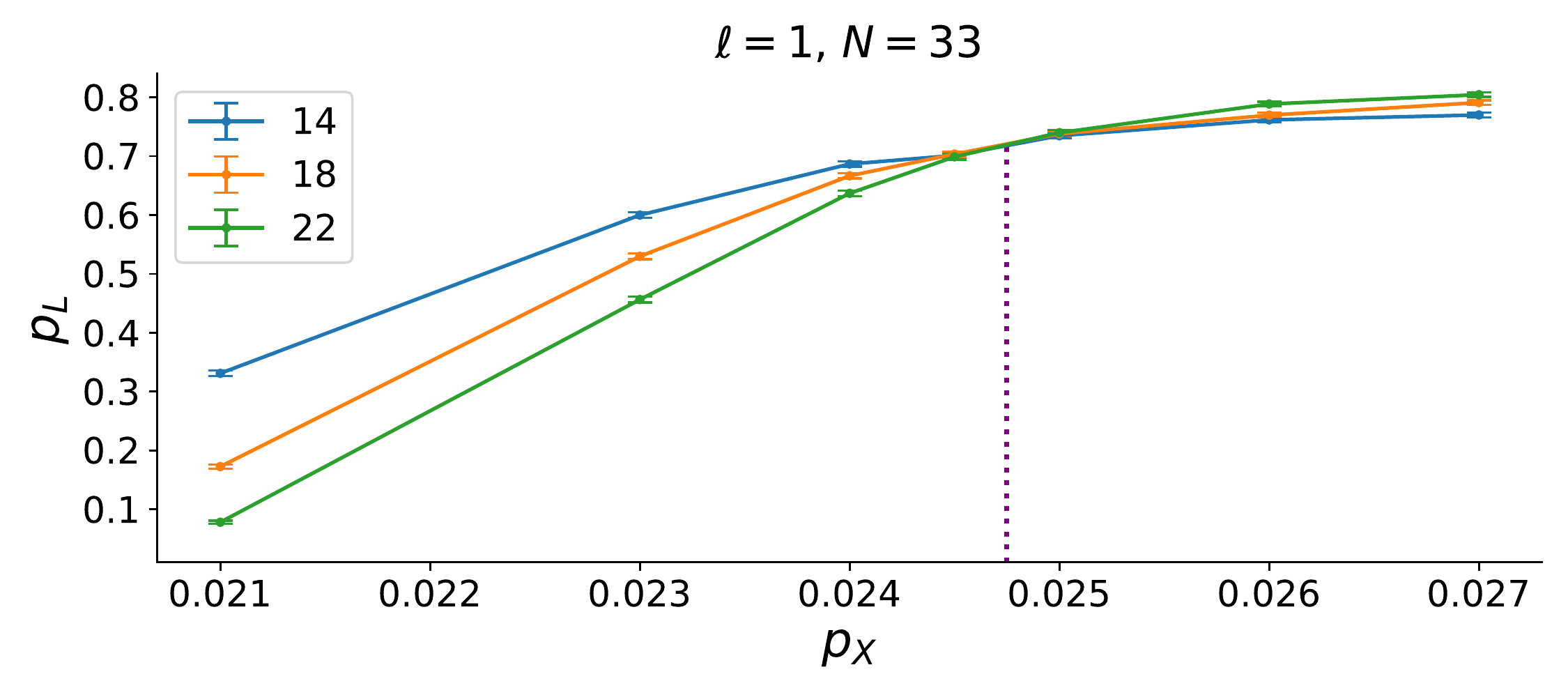}}\\
   \sidesubfloat[]{ \includegraphics[width=0.63\columnwidth]{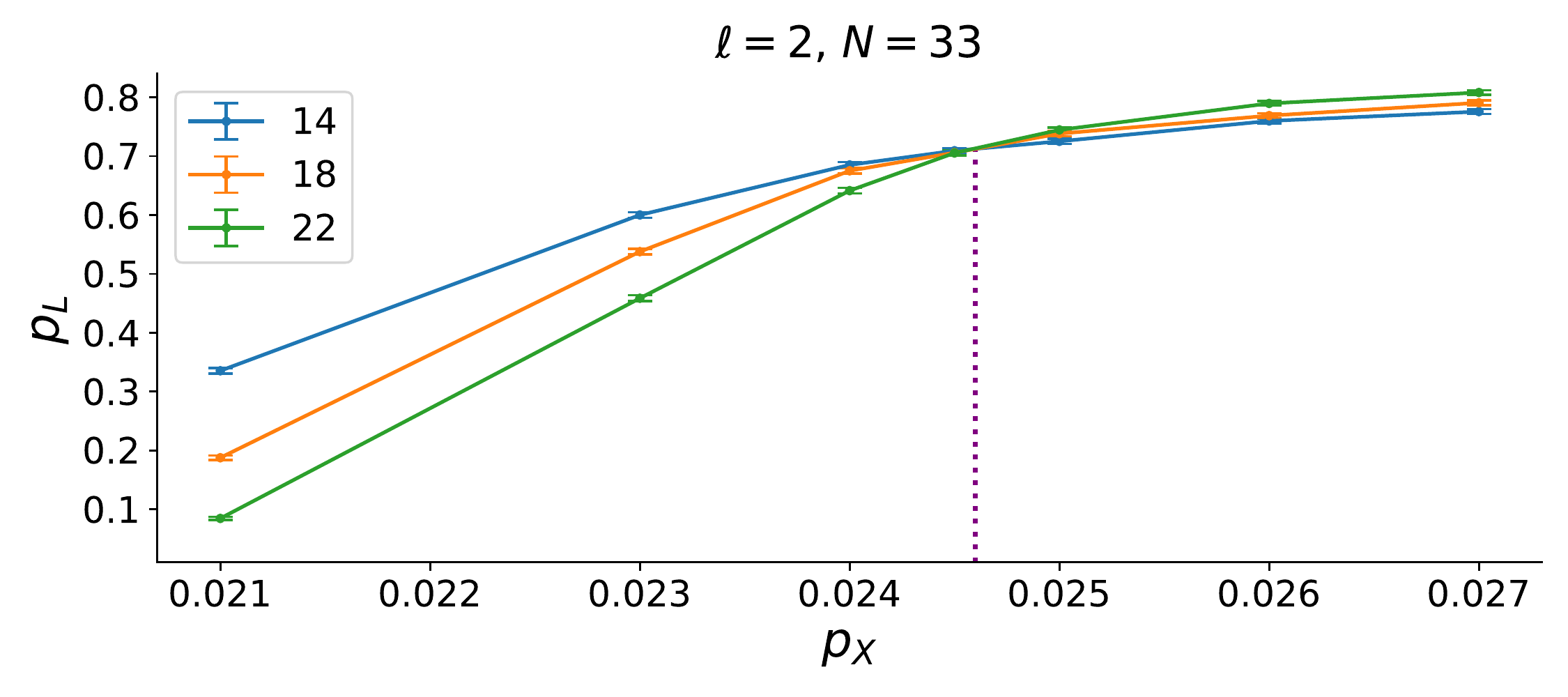}}\\
    \caption{Sweep decoder performance for $N=0,32$ rounds of stabilizer measurements for levels $l=1,2$ of the FSC. Logical failure rate is plotted as a function of the physical error rate $p_X$. The measurement error rate $q$ is set to be same as the physical error rate $p_X$.}
    \label{fig:Sweep_plots_thrN133}
\end{figure}
\bibliographystyle{quantum}
\bibliography{bib,bib2,bibtex_jochym}
\end{document}